\documentclass[11pt]{article}
\usepackage{fullpage}
\usepackage{amsmath}
\usepackage{times}

\usepackage{float}

\newcommand{\comment}[1]{}

\floatstyle{ruled}
\newfloat{algorithm}{h}{loa}
\floatname{algorithm}{Algorithm}

\newtheorem{definition}{Definition}
\newtheorem{theorem}{Theorem}
\newtheorem{corollary}{Corollary}
\newtheorem{lemma}{Lemma}

\newcommand{\qed}{\mbox{\ \ \ }\rule{6pt}{7pt} \bigskip}
\newenvironment{proof}{\par{\bf Proof:}}{\qed \par}

\newenvironment{proofof}[1]{\par{\bf Proof (#1):}}{\qed \par}
\newcommand{\epsapx}{\epsilon}
\newcommand{\epsdist}{\Delta}
\newcommand{\game}{{\cal G}}

\newcommand{\epspert}{\epsilon}

\newcommand{\onem}{{\bf 1}}

\newcommand{\zerom}{{\bf 0}}

\newcommand{\fst}{f}
\newcommand{\poly}{\mathrm{poly}}
\newcommand{\eps}{\epsilon}
\newcommand{\logdel}{\log(1 + \Delta^{-1})}

\usepackage{color}

\newcommand{\X}{\cal X}
\newtheorem{claim}{Claim}

\renewcommand{\Pr}{{\bf Pr}}
\newcommand{\E}{{\bf E}}
\newcommand{\supp}{{\mathrm {supp}}}
\newcommand{\tp}{\tilde{p}}
\newcommand{\Ps}{{\cal P}}
\newcommand{\Qs}{{\cal Q}}

\usepackage[numbers]{natbib}
\begin{document}
\title{\bf  Nash Equilibria in Perturbation Resilient Games}
\author{
Maria-Florina Balcan\footnote{School of Computer Science, College of Computing, Georgia Institute of Technology, Atlanta, Georgia.}
\and Mark Braverman \footnote{Princeton University.}
}
%\renewcommand{\thefootnote}{\fnsymbol{footnote}}
%\addtocounter{footnote}{1}
%\footnotetext{Microsoft Research}
%\addtocounter{footnote}{1}
%\footnotetext{Georgia Institute of Technology, Atlanta, Georgia.}
%\addtocounter{footnote}{-2}

\maketitle

\begin{abstract}
Motivated by the fact that in many game-theoretic settings, the game
analyzed is only an approximation to the game being played, in this work
we analyze equilibrium computation for the broad and natural class of
bimatrix games that are stable to perturbations.  We specifically
focus on games with the property that small changes in the payoff
matrices do not cause the Nash equilibria of the game to fluctuate
wildly.  For such games we show how one can compute approximate Nash
equilibria more efficiently than the general result of Lipton et al.~\cite{LMM03}, by an amount that depends on the
degree of stability of the game and that reduces to their bound in the
worst case.  Furthermore, we show that for stable games the
approximate equilibria found will be close in variation distance to
true equilibria, and moreover this holds even if we are given as input
only a perturbation of the actual underlying stable game.

For uniformly-stable games, where the equilibria fluctuate at most
quasi-linearly in the extent of the perturbation, we get a
particularly dramatic improvement.  Here, we achieve a fully
quasi-polynomial-time approximation scheme: that is, we can find
$1/\poly(n)$-approximate equilibria in quasi-polynomial time.  This is
in marked contrast to the general class of bimatrix games for which
finding such approximate equilibria is PPAD-hard.  In
particular, under the (widely believed) assumption that PPAD is not
contained in quasi-polynomial time, our results imply that
such uniformly stable games are inherently easier for computation of
approximate equilibria than general bimatrix games.
\end{abstract}

\section{Introduction}
The Nash equilibrium solution concept has a long history in economics
and game theory as a description for the natural result of
self-interested behavior~\cite{nash51,gtbook}.  Its importance has led to
significant effort in the computer science literature in recent years
towards understanding their computational structure, and in particular
on the complexity of finding both Nash and approximate Nash
equilibria.  A series of results culminating in the work by Daskalakis, Goldberg, and Papadimitriou~\cite{DGP09} and Chen, Deng, and Teng~\cite{CD06,cdt09}
showed that finding a Nash equilibrium or even a $1/\poly(n)$-approximate
equilibrium, is PPAD-complete even for 2-player bimatrix games. %% ~\cite{DGP09,AKV05,CD06,cdt09}.
%This work has culminated in a series of results showing
%that finding a Nash equilibrium, or even a $1/poly(n)$-approximate
%equilibrium, is PPAD-complete even for 2-player bimatrix games ~\cite{DGP09,AKV05,CD06,cdt09}.
 For general values of $\epsilon$, the best known algorithm for finding $\epsilon$-approximate equilibria
runs in time $n^{O((\log n)/\epsilon^2)}$, based on a structural
result of  Lipton et al.~\cite{LMM03} showing that there always
exist $\epsilon$-approximate equilibria with support over at most
$O((\log n)/\epsilon^2)$ strategies.  This structural result has been
shown to be existentially tight~\cite{FNS07}.  Even for large values of
$\epsilon$, despite considerable effort~\cite{DMP06,DMP07,TS07,FNS07,BBM07,kps06},
polynomial-time algorithms for computing $\epsilon$-approximate
equilibria are known only for $\epsilon \geq 0.3393$~\cite{TS07}. These
results suggest a difficult computational landscape  for
equilibrium and approximate equilibrium computation on worst-case instances.

%The Nash equilibrium solution concept %%%and approximate-Nash equilibria
%has a
%   long history in economics and game theory as a description for the natural result of
%   self-interested behavior~\cite{nash51,gtbook}. In recent years there has been significant effort in the computer science literature on understanding the complexity of finding Nash and approximate Nash equilibria. A series of results showed
%   that  finding a Nash equilibrium is
%  PPAD-complete even for 2-player bimatrix games~\cite{DGP09,AKV05,CD06,cdt09}. As a consequence,
%finding approximate Nash equilibria has emerged
%as the main remaining open question in the area of equilibrium
%computation.  For this solution concept,  %%seminal
%Cheng, Deng, and Teng showed that  it is PPAD-complete to compute an
%  $1/\poly(n)$-approximate equilibria even for two player games~\cite{cdt09}.
%    For general values of $\epsilon$,
%  the best known  algorithms runs in time %%in general
%  is an $n^{O((\log n)/\epsilon^2)}$. This is based on a
%  structural result of Lipton, Markakis, and Mehta~\cite{LMM03} showing that
%  there always exist $\epsilon$-approximate equilibria with support
%  over at most $O((\log n)/\epsilon^2)$ strategies, which  has been shown to be existentially tight~\cite{FNS07}.
%  Even for large values of $\epsilon$, despite considerable effort~\cite{DMP06,DMP07,TS07,FNS07,BBM07,kps06}, polynomial-time
%  algorithms for computing $\epsilon$-approximate equilibria are known
%  only for $\epsilon \geq 0.3393$~\cite{TS07}.

In this paper we  go beyond worst-case analysis and
investigate the equilibrium computation problem in
a natural class of bimatrix games that are stable to perturbations. %%%{\tt slighlty rephraze?!}
As we argue, on one hand, such games %%are quite expressive and
 can be used to model many realistic situations. On the other hand,  we show that they have additional structure  which can be exploited to
 provide better algorithmic guarantees than what is believed to be possible on worst-case instances.
%%%MENTION GOAL???
The starting point of our work is the realization  that games are typically only abstractions of reality and except in the most
controlled settings, payoffs listed in a game that represents an
interaction between self-interested agents are only approximations to
the agents' exact utilities. \footnote{For example, if agents are two corporations with various possible
  actions in some proposed market, the precise payoffs to the
  corporations may depend on specific quantities such as demand for
  electricity or the price of oil, which cannot be fully known in advance but only estimated.
  %As another example, if a company is modeling a business decision as an
%  $n$-action game, the exact payoff values may depend on economic
%  conditions that are available only as forecasts.
%
}
  As a result, for problems such as
equilibrium computation, it is natural to focus attention to games
that are robust to the exact payoff values, in the sense that small
changes to the entries in the game matrices do not cause the Nash
equilibria to fluctuate wildly; otherwise, even if equilibria can be
computed, they may not actually be meaningful for
understanding behavior in the game that is played.
In this work, we focus on such games and analyze their structural
properties as well as their implications to the equilibrium computation
problem. We show how their structure can be leveraged to
obtain better algorithms for computing approximate Nash equilibria,
%%in such games,
as
well as strategies close to true Nash equilibria.  Furthermore, we provide such algorithmic guarantees
 even if we are given only a perturbation of the actual stable game being played.

To formalize such settings we consider bimatrix games $G$ that satisfy what we call the $(\epsilon, \Delta)$ {\em
perturbation stability} condition,  meaning that
for any game $G'$ within $L_\infty$ distance $\epsilon$ of
$G$ (each entry changed by at most $\epsilon$), %%we have that
%the set $S'$ of Nash
%equilibria in $G'$ has Hausdorff distance $\Delta$ from the set
%$S$ of Nash equilibria in $G$ (i.e.,
 each  Nash
equilibrium $(p',q')$ in $G'$ is $\Delta$-close to some
 Nash
equilibrium $(p,q)$ in $G$, where closeness is given by variation distance.
% such that $d((p',q'),(p,q)) \leq \Delta$ where
%$d(.,.)$ is given by variation distance). %%%Different games have different degrees of robustness depending on ...
Clearly, any game is $(\epsilon, 1)$  perturbation stable for any $\epsilon$ and the smaller the $\Delta$ the more structure
the $(\epsilon, \Delta)$ perturbation stable games have. In this paper we study the meaningful range of parameters,
 several structural properties, and the algorithmic behavior of these games.

% approximate Nash equilibria in
%such games.

% as
%well as strategies close to true Nash equilibria.  Furthermore, such
%strategies can be found even if one is given only a perturbation
%$G'$ to the actual game $G$.

%%%provide generic algorithm for computing approximate equilibria whose performance depend on the relationship of these parameters.
Our first main result shows  that for an interesting and general range of parameters
the structure  of  perturbation stable  games can be leveraged to
obtain better algorithms for equilibrium computation. %% in such games.
Specifically, we show that all $n$-action $(\epsapx,\epsdist)$
perturbation-stable
games with at most $n^{O((\Delta/\epsilon)^2)}$ Nash equilibria
must
have a well-supported $\epsapx$-equilibrium of support
$O(\frac{\epsdist^{2}\logdel}{\epsapx^2}\log n)$. This yields an
 $n^{O(\frac{\epsdist^2\logdel}{\epsapx^2} \log
n)}$-time algorithm for finding an $\epsilon$-equilibrium, improving
by a factor $O(1/\epsdist^2\logdel)$ in the exponent over the
bound of \cite{LMM03} for games satisfying this condition (and
reducing to the bound of \cite{LMM03} in the worst-case when
$\epsdist=1$).\footnote{One should think of  $\Delta$ as a function of
$\epsilon$, with both possibly depending on $n$.  E.g.,
$\epsilon = 1/\sqrt{n}$ and $\Delta = 6\epsilon$.}  Moreover,
the stability condition can be further used to show that the
approximate equilibrium found will
be $\Delta$-close to a true equilibrium, and this holds  even if
the algorithm is given as input only a perturbation to the true underlying stable game.

A particularly interesting class of games for which our results provide a
dramatic improvement are those that satisfy what we call {\em
$t$-uniform stability to perturbations}.  These are games that for
some $\epsilon_0 = 1/\poly(n)$ and some $t$ satisfy the $(\epsilon,
t\epsilon)$ stability to perturbations condition for all $\epsilon <
\epsilon_0$.
%A particularly interesting class of games for which this provides a drastic improvement over the previous work are those that satisfy
% what we call the
%  the $(\epsilon,c\epsilon)$
%perturbation-stability condition for all $\epsilon < \epsilon_0$, where $\epsilon_0=1/\poly(n)$  -- we say that such games satisfy the uniformly perturbation stability.
For games satisfying $t$-uniform stability to perturbations with
$t=\poly(\log(n))$, our results imply that we can find
$1/\poly(n)$-approximate equilibria in $n^{\poly(\log(n))}$ time,
i.e., achieve a fully quasi-polynomial-time approximation scheme
(FQPTAS).  This is especially interesting because the results of~\cite{cdt09} prove that it is PPAD-hard to find
$1/\poly(n)$-approximate equlibria in general games.  Our results shows that under the
(widely believed) assumption that PPAD is not contained in
quasi-polynomial time \cite{dask11},  such uniformly stable game are inherently easier for computation of approximate equilibria
than general bimatrix games.\footnote{The generic result of ~\cite{LMM03} achieves
quasi-polynomial time only for $\epsilon = \Omega(1/\poly(\log n))$.}
Moreover, variants of many games appearing commonly in experimental economics including the public goods game, matching pennies, and identical interest game~\cite{gtbook1} %%%, ...{\tt INSERT ONE MORE.... %%rock-paper-scissors, matching pennies, or battle of sexes
satisfy this condition. See Sections~\ref{stable},~\ref{lowerbounds}, and Appendix~\ref{simple-examples} for detailed examples.

%%Moreover,
%%{\tt need to say something about these results being true for any
%%approximation!!!}
%Our upper bounds show that these perturbation stability conditions do impose additional
% structure on the game which could be exploited to provide
%better approximation guarantees.
%This brings up the question of whether such games are significantly simpler. For example, is it significantly easier
%to compute approximate equilibria  in games satisfying the $(\epsapx, \epsapx^{1/c})$ perturbation stability condition, for any constant $c$?
% In this paper we show that this is not the case.

 %%%
Our second main result shows that
 computing an $\epsilon$-equilibrium in a game satisfying  the  $(\epsilon,\Theta(\epsilon^{1/4}))$ perturbation stability condition
 is as hard as computing a $\Theta(\epsilon^{1/4})$-equilibrium in a
general game.
 For our reduction, we show that any general game can be embedded into one having the $(\epsilon,\Theta(\epsilon^{1/4}))$ perturbation stability property such that an $\epsilon$ equilibrium in the new game yields a
 $\Theta(\epsilon^{1/4})$-equilibrium in the original game. This result implies that the interesting range for the  $(\epsapx,\epsdist)$-perturbation stability condition, where one could hope to do significantly better than in the general case,
 is $\epsdist=o(\epsapx^{1/4})$.

We also connect our stability to perturbations condition to a seemingly very different stability to approximations condition introduced by Awasthi et al.~\cite{nash10}. Formally, a game satisfies the   strong $(\epsapx,\epsdist)$-approximation stability condition if all $\epsapx$-approximate
equilibria are contained inside a small ball of radius $\epsdist$
around a single equilibrium.\footnote{~\cite{nash10} argue that this condition is interesting since in situations where
one would want to use an approximate Nash equilibrium for {\em predicting} how players will play (which is a common  motivation
for computing a Nash or an approximate Nash equilibrium), without such
a condition the approximate equilibrium found might
be far from the equilibrium played.}
 We prove that our stability to perturbations condition is equivalent
   to a much weaker version of this condition that we call the
    well-supported approximation stability.  This condition requires
    only that for any well-supported $\epsilon$-approximate equilibrium~\footnote{Recall that
    in an $\epsilon$-Nash equilibrium, the expected payoff of each
player is within $\epsilon$ from her best response
payoff; however the mixed strategies may include poorly-performing
pure strategies in their support. By contrast, the support of a
     well-supported $\epsilon$-approximate equilibrium may only
contain strategies whose payoffs fall within $\epsilon$ of the
player's best-response payoff.}
$(p,q)$ there {\em exists} a Nash
equilibrium $(p^*,q^*)$ that is $\Delta$-close to
$(p,q)$.  Clearly, the well supported  approximation stability
condition is %%substantially
more
general  than strong $(\epsapx,\epsdist)$-approximation stability considered by~\cite{nash10}
since rather than assuming that there exists a
fixed Nash equilibrium $(p^*,q^*)$ such that all
$\epsilon$-approximate equilibria are contained in a ball of radius
$\Delta$ around $(p^*,q^*)$, it requires only that for any
well-supported $\epsilon$-approximate equilibrium $(p,q)$ there exists a Nash
equilibrium $(p^*,q^*)$ that is $\Delta$-close to
$(p,q)$.  Thus, perturbation-stable games are significantly more
expressive than strongly approximation stable games and  Section~\ref{stable} presents several examples of games satisfying
the former but not the latter.
% and thus all our upper bounds apply to the s
%strongly approximation stable games as we
%%%This feels to me like it hurts more than helps so commenting out
%Our upper bound can be seen as significantly generalizing the
% analogue result of Awasthi et al.  Moreover,
However, our lower bound (showing that the interesting  range of parameters is $\epsdist=O(\epsapx^{1/4})$) also holds for the
strong stability to approximations condition.
%Since all our other stability notions are weaker than the strong
%approximation stability condition, this lower bound
% applies to these notions as well.
% Our results are summarized in Table~\ref{fig:summary}.

We also provide an interesting structural result showing that each
Nash equilibrium
%%(and in fact each $\epsilon$-well supported Nash equilibrium)
of an $(\epsilon,\Delta)$ perturbation stable game with
$n^{O((\Delta/\epsilon)^2)}$ Nash equilibria
must be $8\Delta$-close to well-supported
$\epsilon$-approximate equilibrium of support only
$O(\frac{\epsdist^{2}\logdel}{\epsapx^2}\log n)$.  Similarly, a
$t$-uniformly stable game with $n^{O(t^2)}$ equilibria  has the
property that for any $\Delta$, each equilibrium is $8\Delta$-close to
a well-supported $\Delta/t$-approximate equilibrium with support of
size $O(t^2\logdel\log n)$.  This property implies that in
quasi-polynomial time we can in fact find a {\em set} of
approximate-Nash equilibria that cover (within distance $8\Delta$) the
set of {\em all} Nash equilibria in such games.

It is interesting to note that for %%all
our algorithmic results  for finding approximate equilibria we do
{\em not} require  knowing the stability parameters. 
If the game happens to be reasonably stable, then we get improved running times over the Lipton et al.~\cite{LMM03} guarantees; if this is not the case, then we fall back to the  Lipton et al.~\cite{LMM03} guarantees.\footnote{This is because algorithmically, we can simply try different support sizes in increasing order and stop when we find strategies forming a (well-supported) $\epsilon$-equilibrium. In other words, given $\epsilon$, the desired approximation level, we can find an $\epsilon$-approximate equilibrium in time  $n^{O(\frac{\epsdist^2\logdel}{\epsapx^2} \log
n)}$ where $\Delta$ is the smallest value such that the game is $(\epsilon,\Delta)$ perturbation stable. 
% If the game happens to be reasonably stable, then we get improved running times over the results in~\cite{LMM03}. If this is not the case, then we fall back to the guarantees in~\cite{LMM03}. The same holds for $t$-uniform stability to perturbations.
}
However, given a game, it might be interesting to know how stable it is.
In this direction, we provide a characterization of stable constant-sum games in Section~\ref{zero-sum} and
an algorithm for computing the strong stability parameters of a given
constant-sum game.

%%\vspace{-3mm}
\subsection{Related Work}
%%\paragraph{The Complexity of Nash Equilibira}%
%There has been substantial work exploring the computation of Nash
%equilibria in $2$-player $n \times n$ general-sum games.
%%Unfortunately,
%%the complexity results in this area have been almost uniformly
%%negative.
%A series of papers has shown that it is PPAD complete to
%compute Nash equilibria, even in $2$ player games, even when payoffs
%are restricted to lie in $\{0, 1\}$ \cite{DGP09,AKV05,CD06,cdt09}.
%On the positive side, a structural result of Lipton et al.~\cite{LMM03} shows
%that there always exist $\epsilon$-approximate equilibria with support
%over at most $O((\log n)/ \epsilon^2)$ strategies: this gives an
%immediate $n^{O(\log n/ \epsilon^2)}$-time algorithm
%for computing $\epsilon$-approximate
%equilibria.  The best approximation guarantee known for
%polynomial-time algorithms is $0.3393$ for approximate Nash equilibria and ... for well-supported approximate Nash equilibria~\cite{TS07}.
In addition to results on computing (approximate) equilibria in worst-case instances of general bimatrix games,
%here has been substantial work exploring the computation of Nash
%equilibria in $2$-player $n \times n$ general-sum games.
 there has also been a series of results %%%\cite{DMP07,TS07,BBM10}
on polynomial-time algorithms for computing (approximate) equilibria in
specific classes of bimatrix games.
%%%{\tt what is known for well-supported}
%Motivated by this lack of progress, there has  also been recent work analyzing (approximate) equilibria of special classes of
%games.
For example, B{\'a}r{\'a}ny et al.~\cite{BVV07} considered two-player games with
randomly chosen payoff matrices, and showed that with
high probability, such games have Nash equilibria with
small support. Their result implies that in random
two-player games, Nash equilibria can be computed in expected
polynomial time. Kannan and  Theobald~\cite{kt10} provide an FPTAS for the case where
the sum of the two payoff matrices has constant rank and Adsul et al.~\cite{ruta11} provide a polynomial time algorithm for computing an exact Nash equilibrium
of a rank-1 bimatrix game.

Awasthi et al.~\cite{nash10}  analyzed the question of
finding an approximate Nash in equilibrium in games that satisfy %a natural notion of
stability with respect to approximation.
However, their condition is quite restrictive in that it focuses only on games that have the property that all the Nash equilibria are close together, thus eliminating from consideration most common games. By contrast, our perturbation stability notion, which (as mentioned above) can be shown to be a generalization of their notion, captures many more realistic situations. Our upper bounds on approximate equilibria can be viewed as generalizing the corresponding result of~\cite{nash10} and it is  significantly more challenging technically. Moreover, our lower bounds
%%on the useful range of parameters from the point of view of the equilibrium computation problem
also apply to the stability notion of~\cite{nash10} and provide  the  first (nontrivial) results about the interesting range of parameters for that stability notion as well.
%a consequence of our work is that
%we significantly extend their results, by both clarifying the interesting range of parameters for the
%approximation-stability property and by extending their general upper bound to a much more compelling  and broader notion of
%stability.

In a very different context, for clustering problems, Bilu and Linial~\cite{BL} analyze maxcut clustering instances with the property that if the distances are perturbed by a multiplicative factor of $\alpha$, then the optimum does not change; they show that  under this condition, for $\alpha=\sqrt{n}$, one can find the optimum solution in polynomial time. Recently, Awasthi et al.~\cite{ABS11}, have shown a similar result for
the k-median clustering problem and showed a similar result for $\alpha=\sqrt{3}$.   Our stability to perturbations notion is inspired by this work, but is substantially less restrictive in two respects.  First, we require
  stability only to small perturbations in the input, and second, we
  do not require the solutions (Nash equilibria) to stay fixed under
  perturbation, but rather just ask that they have a bounded degree of
  movement.

The notion of stability to perturbations we consider in our paper is also related to the stability notions considered by Lipton et al.~\cite{Vangelis} for economic solution concepts.  The main focus of their work was understanding whether for a given solution
concept or optimization
problem {\em all} instances are stable. In this paper, our main focus is on understanding how rich the class
of stable instances is, and what properties one can determine about
their structure that can be leveraged to get better algorithms for computing approximate Nash equilibria.\footnote{Just as in~\cite{Vangelis}, one can show that for the
stability conditions we consider in our paper, there exist unstable instances.}
We provide the first results showing better algorithms for computing approximate equilibria in such games.

%\paragraph{Structure of this paper} We start with preliminaries in Section~\ref{defs}. We then formally introduce the
%stability notions we analyze and discuss relationships between them in Section~\ref{stable}. We provide
%a general upper bound on the support of stable games  and discuss algorithmic implications in Section~\ref{stable-upper}.
%We present a general reduction from general games to stable games in Section~\ref{lowerbounds}.
% We also provide a characterization of stable constant-sum games in Section~\ref{zero-sum}.
%% and conclude with a discussion and open questions
%%%in Section~\ref{conclusions}.

\section{Preliminaries}
\label{defs}
We consider 2-player general-sum $n$-action bimatrix games.  Let $R$
denote the payoff matrix to the row player and $C$ denote the payoff
matrix of the column player.   If the row player chooses
strategy $i$ and the column player chooses strategy $j$,
the payoffs are $R_{i,j}$ and $C_{i,j}$ respectively.
We assume all payoffs are scaled to the
range $[0,1]$.

A mixed strategy for a player is a probability distribution over the set of his pure strategies.
%and will be represented by a vector $p = (p_1, p_2, ..., p_n)^T$ , where $p_i \geq 0$ and
%$\sum_{i=1}^{n} p_i = 1$. Here $p_i$ is the probability that the player will choose his $i$th pure strategy.
The $i$th pure strategy will be represented by the unit vector $e_i$, that has $1$ in the $i$th coordinate and $0$ elsewhere.
For a mixed strategy pair $(p,q)$, the payoff to the row player is the expected value of a random
variable which is equal to $R_{i,j}$ with probability $p_iq_j$. Therefore the payoff to the row player is
$p^T R q$. Similarly the payoff to the column player is $p^TC q$. Given strategies $p$ and $q$  for the row and column player,
we denote by $\supp(p)$ and $\supp(q)$ the support
of $p$ and $q$, respectively.
%%%THOSE STRATEGIES WITH P_I>0

A Nash equilibrium~\cite{nash51} is a pair of strategies $(p^*, q^*)$ such that no player has an incentive
to deviate unilaterally. Since mixed strategies are convex combinations of pure strategies, it
suffices to consider only deviations to pure strategies.  In particular,  a pair of mixed strategies $(p^*,q^*)$ is a
Nash-equilibrium if for every pure strategy $i$ of the row player we have
 $e_i^T R q^* \leq {p^*}^TR q^*$, and for every pure strategy $j$ of the column player we have
 ${p^*}^TC e_j \leq {p^*}^TCq^*$.
Note that in a Nash
equilibrium  $(p^*,q^*)$, all rows $i$ in the support of $p^*$ satisfy
$e_i^TRq^* = {p^*}^T Rq^*$ and similarly all columns $j$ in the support
of $q^*$ satisfy ${p^*}^TCe_j = {p^*}^TCq^*$.

\begin{definition}
A pair of mixed strategies $(p,q)$ is an
$\epsilon$-equilibrium if both players have no more than $\epsilon$
incentive to deviate.  Formally, $(p,q)$ is an $\epsilon$-equilibrium
if for all rows $i$, we have $e_i^TRq \leq p^TRq+\epsilon$, and for all
columns $j$, we have $p^TCe_j \leq p^TCq + \epsilon$.
\end{definition}

\begin{definition}
A pair of mixed strategies $(p,q)$ is a well supported
$\epsilon$-equilibrium if for any $i \in \supp(p)$ (i.e., $i$ s. t. $p_i >0$) we have
$e_i^TRq \geq e_j^T Rq - \epsilon$, for  all $j$; similarly, for any $i \in \supp(q)$ (i.e., $i$ s. t. $q_i >0$) we have
$p^T C e_i \geq p^T C e_j - \epsilon$, for  all $j$.
\end{definition}

\begin{definition}
We say that a bimatrix game $G'$ specified by $R', C'$ is an $L_\infty$ $\alpha$-perturbation of $G$ specified by $R, C$ if we have
$|R_{i,j} - R'_{i,j}| \leq \alpha$ and
$|C_{i,j} - C'_{i,j}| \leq \alpha$ for all $i,j \in \{1, \ldots, n\}$.
\end{definition}

\begin{definition}
For two probability distributions  $q$ and $q'$, we define the distance between $q$ and $q'$ as the variation distance:
$$d(q,q') = \frac{1}{2}\sum_i |q_i-q'_i| = \sum_i \max(q_i-q'_i,0) =
\sum_i \max(q'_i-q_i,0).$$
%%It is easy to see that $d$ is a metric.
%Following~\cite{nash10} we define
We  define the distance between two strategy pairs as the maximum
of the row-player's and column-player's distances, that is:
$d((p,q),(p',q')) = \max[d(p,p'), d(q,q')].$
\end{definition}
%%Again, it is easy to see that $d$ is a metric.
It is easy to see that $d$ is a metric.
%\footnote {I.e., $d((p,q),(p',q'))=0$ iff $(p,q)=(p',q')$, $d((p,q),(p',q'))=d((p',q'),(p,q))$, and
%$d((p,q),(p'',q'')) \leq d((p,q),(p',q'))+ d((p',q'),(p'',q''))$, for any $(p,q)$,$(p',q')$, $(p'',q'')$.}
If $d((p,q),(p',q')) \leq \Delta$, then we say that $(p',q')$ is $\Delta$-close to $(p,q)$.

Throughout this paper we use ``log'' to mean log-base-e.

\section{Stable Games}
\label{stable}
%In this section we present several natural notions of stable games.
%We start with a notion of approximate stability  based on the notion in~\cite{nash10}.
The main notion of stability we introduce and study in this paper requires that any Nash equilibrium in a perturbed game be close to a Nash equilibrium in the original game. This is an especially motivated condition since in many real world situations the  entries of the game we analyze are merely based
on measurements and thus  only approximately reflect the players' payoffs.
In order to be useful for prediction, we would like that equilibria in the game we operate with  be close to  equilibria in the real game played by the players. Otherwise, in games where  certain equilibria of slightly perturbed games are far from all equilibria in the original game,  the analysis of behavior (or prediction) will be meaningless.
 Formally:

\begin{definition}
\label{stab3}
A game $G$ satisfies the $(\epspert, \epsdist)$  stability to perturbations condition  if for any $G'$ that is an
$L_\infty$ $\epspert$-perturbation of $G$ and for any Nash
equilibrium $(p,q)$ in $G'$, there exists a Nash equilibrium
$(p^*,q^*)$ in $G$ such that $(p,q)$ is $\Delta$-close
to $(p^*,q^*)$.
\footnote{Note that the entries of the perturbed game are not restricted to the $[0,1]$ interval, and are allowed to belong to $[-\epspert, 1+ \epspert]$. This is a proper way to
formulate the notion because it implies, for instance, that
if $G$ is $(\epsilon,\Delta)$ stable to perturbations, then for any
$\alpha>0$, $\alpha G$ is $(\alpha\epsilon,\Delta)$ stable to
perturbations. Theorem~\ref{stab3eqstab2} provides further evidence that this definition is proper.}

%for any way to modify the entries up to an additive factor of $\epspert$, there exists a Nash equilibrium $(p^*,q^*)$ such that any Nash equilibria in the new game is $\epsdist$ close to
%$(p^*,q^*)$.
%
%A game G satisfies the strong $(\epspert, \epsdist)$  stability to perturbations condition  if there exists a Nash equilibrium $(p^*,q^*)$ such that
%for any way to modify the entries up to an additive factor of $\epspert$ any Nash equilibria in the new game is $\epsdist$ close to
%$(p^*,q^*)$.
\end{definition}

%As mentioned in Section~\ref{intro}, this stability notion is related to the one considered by Bilu and Linial in a recent paper in the context of clustering~\cite{BL}, as well as to the stability notion of Lipton,
%Markakis, and Mehta~\cite{Vangelis}.

Observe that fixing $\epsapx$, a smaller $\epsdist$ means a stronger condition
and a larger $\epsdist$ means a weaker condition.  Every game is
$(\epsapx,1)$-perturbation stable, %%(or approximation stable),
and as $\epsdist$ gets smaller, we
might expect for the game to exhibit more useful structure.

Another stability condition we consider in this work is stability to approximations:
%We now present an interesting and useful relationship between our stability to perturbations notion and the stability to approximations notion considered in prior work. We start with describing the  stability to perturbations condition. %%introduced by Awasthi et al.~\cite{nash10}.

\begin{definition}
\label{stab1}
A game satisfies  the $(\epsapx, \epsdist)$-approximation
stability condition if for any $\epsapx$-equilibrium  $(p,q)$  there
exists a Nash equilibrium $(p^*,q^*)$ such that $(p,q)$  is $\epsdist$-close to $(p^*,q^*)$.
%%that is $d((p,q),(p^*,q^*)) \leq \epsdist$.

A game satisfies the well supported $(\epsapx, \epsdist)$-approximation
stability condition if  for any  well supported $\epsapx$-equilibrium $(p,q)$ there exists a Nash equilibrium $(p^*,q^*)$ such that
$(p,q)$ is $\epsdist$-close to $(p^*,q^*)$.
%%i.e. $d((p,q),(p^*,q^*)) \leq \epsdist$.
\end{definition}

Clearly, if a game satisfies the  $(\epsapx, \epsdist)$-approximation
stability condition, then it also satisfies the well supported $(\epsapx, \epsdist)$-approximation
stability condition.  Interestingly, we show that the stability to perturbations  condition is equivalent to the
well supported approximation stability condition.  Specifically:

\begin{theorem}
\label{stab3eqstab2}
A game satisfies the well supported $(2\epsapx, \epsdist)$-approximation stability condition if and only if it satisfies the $(\epsapx, \epsdist)$-stability to perturbations  condition.
\end{theorem}

\begin{proof}
%We first show that if a game satisfies the  well supported $(2\epsapx, \epsdist)$-approximation stability condition,
%then it also satisfies the $(\epsapx, \epsdist)$-stability to perturbations  condition.
Consider an $n \times n$ bimatrix game specified by $R$ and $C$ and assume  it satisfies the well supported $(2\epsapx, \epsdist)$-approximation stability condition;
we show it also satisfies the $(\epsapx, \epsdist)$-stability to perturbations  condition.
 Consider $R'=R+\Gamma$ and $C'=C+\Lambda$, where  $|\Gamma_{i,j}|\leq \epsilon$ and $|\Lambda_{i,j}|\leq \epsilon$, for all $i,j$.
Let $(p,q)$  be an arbitrary  Nash equilibrium in the new game specified by $R'$ and $C'$.
  We will show that $(p,q)$ is a well supported $2\epsapx$-approximate Nash equilibrium in the original game specified by $R$ and $C$.
To see this, note that by definition, (since $(p,q)$ is a Nash equilibrium in the game specified by $R'$ and $C'$) we have
$e_j^T R' q \leq  p^T R' q \equiv v_R$ for all $j$; therefore
$e_j^T R q + e_j^T \Gamma q \leq  v_R $, so $e_j ^T R q \leq  v_R + \epsapx$, for all $j$. On the other hand we also have
$e_i^T R q =  e_i^T R' q  -  e_i^T \Gamma q \geq v_R - \epsapx$ for all $i \in \supp(p)$.
Therefore, $e_i^T R q \geq e_j^T R q - 2\epsapx$, for all $i \in \supp(p)$ and for all $j$.
Similarly we can show $p^T C e_i \geq p^T C e_j - 2\epsapx$, for all $i \in \supp(q)$ and for all $j$.
This implies that  $(p,q)$ is well supported $2 \epsapx$-approximate
Nash in the original game, and so by assumption
is $\epsdist$-close to a  Nash equilibrium of the game
  specified by $R$ and $C$.  So, this game satisfies the $(\epsapx,
\epsdist)$-stability to perturbations  condition.

In the reverse direction, consider an $n \times n$ bimatrix game specified by $R$ and $C$ and assume it satisfies the  $(\epsapx, \epsdist)$-stability to perturbations  condition.
Let $(p,q)$  be an arbitrary well supported $2\epsapx$ Nash equilibrium in this game. Let us define matrices $R'$ and $C'$ such that
$e_i^T R' q =\max_{i'} {e_{i'}}^T R q-\epsilon$ for all $i \in \supp(p)$ and $e_i^T R' q \leq \max_{i'} e_{i'}^T R q-\epsilon$  for all $i \notin \supp(p)$, $p^T C' e_j =\max_{j'} p^T C e_{j'}-\epsilon$ for all $j \in \supp(q)$ and  $p^T C' e_{j'} \leq \max_{j} p^T C e_{j'} -\epsilon$ for all $j \notin \supp(q)$.
Since  $(p,q)$  is a well supported $2\epsapx$ Nash equilibrium we know this can be done such that $|(R'-R)_{i,j}| \leq \epsilon$ and
$|(C'-C)_{i,j}| \leq \epsilon$, for all $i,j$ (in particular, we have to add quantities in $ [-\epsilon,\epsilon]$ to all
 the elements in rows $i$ of $R$ in the support of $p$ and subtract quantities in $[0,\epsilon]$ from all
 the elements in rows $i$ of $R$ not in the support of $p$; similarly for $q$).  By design, $(p,q)$ is a Nash equilibrium in the game defined by $R'$, $C'$, and from the
$(\epsapx, \epsdist)$-stability to perturbations  condition, we obtain that it is $\Delta$-close to a true Nash equilibrium of the game specified by $R$ and $C$.
Thus, any  well supported $2\epsapx$ Nash equilibrium in the game specified by $R$ and $C$ is $\Delta$-close to a true Nash equilibrium of this game, as desired.
\end{proof}

%\vspace{-3mm}

One can show that the well supported approximation stability is a strict relaxation of the approximation stability condition.
%%In fact, the later is a strict relaxation of the former.
For example,  consider the  $2 \times 2$  bimatrix game %%with %%defined by The row and the column matrices are $2 \times 2$ as follows:
%\begin{align*}
$$R = \left[\begin{array}{cc}
           1 & 1 \\

           1-\epsilon_0 & 1-\epsilon_0
\end{array}\right]\hspace{5mm}
C = \left[\begin{array}{cc}
           1 & 1-\epsilon_0 \\
           1 & 1-\epsilon_0
\end{array}\right]
$$

For $ \epsapx< \epsilon_0$ this game satisfies the well supported
$(\epsapx, 0)$-approximation stability condition, but does not satisfy
$(\epsapx, \Delta)$-approximation stability for any $\Delta <
\epsapx/\epsilon_0$.
%%% (see Appendix~\ref{simple-examples}).
To see this note that $e_1^T R q=1$,  $e_2^T R q=1-\epsilon_0$,  $p^T C e_1=1$, and $p^T C e_2=1-\epsilon_0$ for any $p$ and $q$. This implies that the
only well supported $\epsilon$-Nash equilibrium is identical to the Nash equilibrium $(1,0)^T, (1,0)^T$,
thus the game is well supported $(\epsapx,0)$-approximation stable. %% for any $ \epsapx< \epsilon_0$.
On the other hand, the pair of mixed strategies $(p,q)$ with $p=(1-\epsapx/\epsilon_0, \epsapx/\epsilon_0)^T$
and $q=(1-\epsapx/\epsilon_0, \epsapx/\epsilon_0)^T$ is an $\epsilon$-Nash equilibrium.
%%  since $p^T R q= 1-\epsilon$ and $p^T C q =1-\epsilon$.
  The  distance between $(p,q)$ and the unique Nash is
$\epsapx/\epsilon_0$, thus  this game is
not $(\epsapx,\Delta)$-approximation stable for any
$\Delta<\epsapx/\epsilon_0$.

Interestingly, the notion of approximation stability which is a %%strict
restriction of the well-supported approximation stability and of stability to perturbations  conditions is a relaxation of the stability condition considered by Awasthi et al.~\cite{nash10} which requires that all approximate equilibria  be contained in a ball of radius $\Delta$ around a single Nash   equilibrium. In this direction, we define the {\em strong version} of stability
  conditions given in Definitions 3 and 4 to be a reversal of
  quantifiers that asks there be a single $(p^*,q^*)$ such that each
  relevant $(p,q)$ (equilibrium in an $\epsilon$-perturbed game,
  $\epsilon$-approximate equilibrium, or well-supported
  $\epsilon$-approximate equilibrium) is $\Delta$-close to $(p^*,q^*)$. Formally:

\begin{definition}
\label{stab3}
A game $G$  satisfies the strong $(\epspert, \epsdist)$  stability to perturbations condition  if there exists $(p^*,q^*)$  a Nash equilibrium  of $G$ such that
for any $G'$ that is an
$L_\infty$ $\epsilon$-perturbation of $G$ we have that any Nash
equilibrium in $G'$ is $\Delta$-close
to $(p^*,q^*)$.

A game $G$ satisfies  the strong (well supported) $(\epsapx, \epsdist)$-approximation
stability condition if there exists $(p^*,q^*)$  a Nash equilibrium  of $G$  such that  any (well supported) $\epsapx$-equilibrium  $(p,q)$  is $\epsdist$-close to $(p^*,q^*)$.
%%%%%that is $d((p,q),(p^*,q^*)) \leq \epsdist$.
\end{definition}

It is immediate from its proof that Theorem~\ref{stab3eqstab2} applies to the
strong versions of the definitions as well. We also note that our generic upper bounds in Section~\ref{stable-upper} will apply to the most
  relaxed version (perturbation-stability) and  our generic lower bound in Section~\ref{lowerbounds}
  will be apply to the most stringent version (strong approx stability).

%
%We finally present relations between
%  $\epsilon$ and $\Delta$ that must hold for these definitions in Appendices B and  C.
 \paragraph{Range of parameters} As shown in~\cite{nash10},  if a game $\game$ satisfies the strong $(\epsapx,
\epsdist)$-approximation stability  condition and has a non-trivial Nash
equilibrium (an equilibrium in which the players do not both have full support), then we must have  $ \epsdist \geq  \epsapx$.
We can show that if a game $\game$ satisfies the $(\epsapx,
\epsdist)$-approximation stability and if the union of all $\Delta$-balls around all Nash equilibria does not cover the whole space,\footnote{If the union of all $\Delta$-balls around all Nash equilibria does  cover the whole space, this is an easy case from our perspective. Any $(p,q)$ would be a $\epsilon$-equilibria.} then we must have  $ 3\epsdist \geq  \epsapx$ -- see Lemma~\ref{parameters} in Appendix~\ref{proofs}.
In Section~\ref{lowerbounds} we further discuss the meaningful range of parameters from the point of view of the equilibrium computation problem.

\paragraph{Examples}
%Throughout this paper we will analyze  structural properties and algorithmic implications of 2-player general-sum $n$ action bimatrix games.
%We conclude this section with a couple of example satisfying
 Variants of  many classic games including the public goods game, matching pennies, and identical interest games are stable.
As an example, consider the following modified identical interest game. Both players have  $n$ available actions. The first action is to stay home, and the other actions  correspond to $n-1$ different possible meeting locations. If a player chooses action $1$ (stay home), his payoff is $1/2$ no matter what the other player is doing. If the player chooses to go out to a meeting location, his payoff is $1$ if the other player is there as well and  it is $0$ otherwise. This game has  $n$ pure equilibria (all $(e_i,e_i)$) and ${n \choose 2}$ equilibria (all $(1/2 e_i + 1/2 e_j, 1/2 e_i + 1/2 e_j)$)
and it is well-supported $(\epsilon,2\epsilon)$-approximation
   stable for all $\epsilon < 1/6$.  Note that it does not satisfy
   strong  (well-supported)  stability because it has multiple very distinct equilibria.
%%and it is $(\epsilon,2\epsilon)$-stable for all $\epsilon < 1/6$.
For further examples see Lemma~\ref{basic-game-stable} in Section~\ref{lowerbounds}, as well as Appendix~\ref{simple-examples}.

%Consider the public goods game defined as follows. We have two players and each can choose to play a number between $0$ and $n-1$. If the Row player contributes $i$ dollars and Column player contributes $j$ dollars, then each gets back $0.75(i+j)$.  %%I am assuming we multiply by 1.5 what both players put down and divide by 2.
%So the payoff to the Row player is $0.75(i+j)-i$ and the payoff to the Column player is $0.75(i+j)-j$,
%where $i \in \{0,1,...,n-1\}$ and $j \in \{0,1,...,n-1\}$.
%This has payoffs ranging from $0$ up to $0.75(n-1)$, so to scale to the range $[0,1]$ as we do in our paper, we multiply all the payoffs by $1/n$.  I.e., if the Row player plays i and the Column player plays j then the payoff to the Row player is- $[0.75j-0.25i]/n$ and the payoff to the Column player is
%$[0.75 i-0.25 j]/n$. This game is $(\epsilon,0)$ stable to perturbations for all $\epsilon < 1/(8n)$, but is  not $(\epsilon,0.99)$ stable for any $\epsilon > 1/(4n)$.

\section{Equilibria in Stable Games}
\label{stable-upper}
In this section we show we can leverage the structure implied by stability to perturbations  to improve over the best known generic bound of \cite{LMM03}.
%%Clearly, using the equivalence in Theorem~\ref{stab3eqstab2},  it is sufficient to show this for well supported approximation stable games.
We start by considering $\epsilon$ and $\Delta$ as given.
%%\footnote{It is best to think of  $\Delta$ as a function of $\epsilon$, with both possibly depending on the dimension.}.
We can show:

%% The proof uses some ideas from~\cite{nash10} for the strong approximation stability condition;
%%however, we it is technically much more challenging since  we deal here with well supported approximate
%% Nash equilibria as opposed to approximate Nash equilibria as in~\cite{nash10}, and more importantly, we allow for games with
%%$n^{O((\epsdist/\epsapx)^2)}$ equilibria.

\begin{theorem}\label{thm:weak:small-support}
Let us fix $\epsilon$ and $\Delta$, $0\leq \epsilon \leq \Delta \leq 1$.
Consider a game with at most $n^{O((\epsdist/\epsapx)^2)}$ Nash equilibria which satisfies  the well supported $(\epsapx, \epsdist)$-approximation stability condition (or the  $(\epsapx/2, \epsdist)$-stability to perturbations  condition).
Then
there exists a well supported $\epsapx$-equilibrium  where each player's strategy has
support $O((\epsdist/\epsapx)^2\logdel\log n)$.
\end{theorem}

This improves
by a factor $O(1/(\epsdist^2\logdel))$ in the exponent over the
bound of \cite{LMM03} for games satisfying these conditions (and
reduces to the bound of \cite{LMM03} in the worst-case when
$\epsdist=1$).

\paragraph{Proof idea} We start by  showing that any Nash equilibrium
$(p^*,q^*)$ of $G$ must be highly concentrated.  In particular, we
show that for each of $p^*,q^*$, any portion of the distribution with
substantial $L_1$ norm (having total probability at least $8\Delta$)
must also have high $L_2$ norm: specifically the ratio of $L_2$ norm
to $L_1$ norm must be $\Omega((\epsilon/\Delta)(\log n)^{-1/2})$.
This in turn can be used to show that each of $p^*,q^*$ has all but at
most $8\Delta$ of its total probability mass is concentrated in a set (which we call the high part) of
size $O((\Delta/\epsilon)^2 \log(1 + 1/\Delta)\log n)$.  Once the desired concentration is proven, we can then perform a version of the~\cite{LMM03}
sampling procedure on the low parts of $p^*$ and $q^*$ (with an accuracy of only $\epsilon/\Delta$) to produce overall an
$\epsilon$-approximate equilibrium of support only a constant factor
larger.  The primary challenge in this argument is to prove that
$p^*$ and $q^*$  are concentrated.\footnote{We note that ~\cite{nash10} prove an upper bound  for the strong approximation stability condition using the same concentration idea. However, proving the desired concentration is significantly more challenging in our case since we deal with many equilibria.}  This is done through our
key lemma, Lemma~\ref{clm:p-small-norm} below.  In particular, Lemma~\ref{clm:p-small-norm}
can be used to show that if $p^*$ (or $q^*$) had a portion with substantial $L_1$ norm
%%probability mass
 and low $L_2$ norm, then there must exist a
deviation from $p^*$ (or $q^*$) that is far from {\em all} equilibria and yet is
a well-supported approximate-Nash equilibrium, violating the stability
condition.  Proving the existence of such a deviation is challenging
because of the large number of equilibria that may exist.
Lemma~\ref{clm:p-small-norm} synthesizes the key points of this argument and it is
proven through a careful probabilistic argument.

%% that is based on an idea appearing~\cite{nash10} as well as the grand challenge we have to overcome.
%%%%It shows that ...
%%We start with a key technical probabilistic lemma.

In the following  we consider $c =(56)^2$ and let $c'=27$.
\begin{lemma}\label{clm:p-small-norm}
Let us fix $\epsilon$ and $\Delta$, $0\leq \epsilon \leq \Delta \leq 1$.
Let $\tp$ be an arbitrary distribution over $\{1,2, \ldots, n\}$. Let $S = c(\epsdist/\epsapx)^2\log n$ %% for sufficiently large constant $c$
and fix $\beta \leq 1$ such that  $1-\beta \geq 8\Delta$.  Assume that the entries of $\tp$ can be partitioned into two sets $H$ and $L$  such that $\|\tp_L\|_1=1-\beta$, $\|\tp_H\|_1=\beta$, $\|\tp_L\|_2^2 \leq \frac{(1-\beta)^2}{S}$.
Let us fix $k_1$ n-dimensional vectors $v_1, \ldots, v_{k_1}$ with entries in $[-1,1]$ and $k_2$ distributions $p_1, \ldots, p_{k_2}$, where $k_1=n^2$ and $k_2 \leq n^{c'(\epsdist/\epsapx)^2}$. Then there exists $\tp'$  with $\supp(\tp') \subseteq \supp(\tp)$ such that:
\begin{enumerate}
\item $d(\tp,\tp') = 3\Delta$ and
\item $\tp' \cdot v_i \leq \tp \cdot v_i + \epsilon$ for all $i \in \{1, \ldots, k_1\}$.
\item $d(\tp',p_i) > d(\tp,p_i)- \Delta$ for all $i \in \{1, \ldots, k_2\}$..
\end{enumerate}
\end{lemma}

\begin{proof}
We show the desired result by using the probabilistic method.
Let us define the random variable $X_i=1$ with probability $1/2$ and $X_i=0$ with probability $1/2$.
Define $$\tp'_i=\tp_i,~~~~~~~\mathrm{for}~~~~~~i \in H~~~~~~\mathrm{and}~~~~~~~~~~~~~~~~~~~~~~~~~~~~~~~~~~~~~~~$$
 $$\tp'_i=\tp_i + \frac{3 \Delta \tp_i X_i}{\sum_{i \in L} \tp_i X_i } - \frac{  3 \Delta \tp_i(1- X_i)}{\sum_{i \in L} \tp_i (1-X_i) },~~~\mathrm{for}~~~i \in L.$$
We have $E[\sum_{i \in L} \tp_i X_i]=(1-\beta)/2$. By applying McDiarmid's inequality (see Theorem~\ref{mcd}) and using the fact that $\|\tp_L\|_2^2 \leq \frac{(1-\beta)^2}{S}$, we obtain that with probability at least $3/4$ we have both:
\begin{eqnarray}|\sum_{i \in L} \tp_i X_i -\frac{1-\beta}{2}| \leq \frac{1-\beta}{12}~~~~~\mathrm{and}~~~~~|\sum_{i \in L} \tp_i (1-X_i) - \frac{1-\beta}{2}| \leq \frac{1-\beta}{12}.\label{denom}\end{eqnarray}

Assume that this happens. In this case, $\tp'$ is a legal mixed strategy for the row player  and by construction we have $d(\tp, \tp')=3 \Delta$.

Let $v$ be an arbitrary vector in $\{v_1, \ldots, v_{k_1}\}$.
We have:
$$
\tp' \cdot v = \tp\cdot v + 3\Delta\left(
                   \frac{\sum_{i \in L}\tp_i X_i v_i}
			{\sum_{i \in L}\tp_i X_i}
                 - \frac{\sum_{i \in L}\tp_i (1-X_i) v_i}
			{\sum_{i \in L}\tp_i (1-X_i)} \right).
$$
%Our aim is to show that with high probability, $\tp'\cdot v \leq
%\tp\cdot v + \epsilon$.
Define
$$
Z_1 \; = \; \sum_{i \in L}\tp_i X_i v_i,\;\;\;
Z_2 \; = \; \sum_{i \in L}\tp_i X_i,\;\;\;
Z_3 \; = \; \sum_{i \in L}\tp_i (1-X_i) v_i, \;\;\;
Z_4 \; = \; \sum_{i \in L}\tp_i (1-X_i),
$$
so we have:
$$
\tp'\cdot v = \tp\cdot v + 3\Delta\left(\frac{Z_1}{Z_2} -
\frac{Z_3}{Z_4}\right).
$$

Using McDiarmid's inequality we get that with probability at least $1-1/n^3$, each
of the quantities $Z_1,Z_2,Z_3,Z_4$ is within
$(\frac{1-\beta}{28})(\frac{\epsilon}{\Delta})$ of its expectation; we
are using here the fact that the value of $X_i$ can change any one of
the quantities by at most $\tp_i$, so the exponent in the
McDiarmid bound is $-(\frac{\epsilon}{\Delta})^2(\frac{1-\beta}{28})^2 / \sum_{i\in
L}\tp_i^2 \leq -(\frac{c}{28^2})\log n$.  Also, we have $\E[Z_2] =
\E[Z_4] = \frac{1-\beta}{2}$, and $\E[Z_1] = \E[Z_3]$.
So, we get that with probability at least $1-1/n^3$ we have
\begin{eqnarray*}
\tp'\cdot v & \leq &
   \tp\cdot v + 3\Delta\left(
     \frac{\E[Z_1] + \frac{(1-\beta)\epsilon}{28\Delta}}
	  {\frac{1-\beta}{2} - \frac{(1-\beta)\epsilon}{28\Delta}}
      -
     \frac{\E[Z_3] - \frac{(1-\beta)\epsilon}{28\Delta}}
	  {\frac{1-\beta}{2} + \frac{(1-\beta)\epsilon}{28\Delta}}
     \right)\\
& = &
   \tp\cdot v + 3\Delta\left(
     \frac{(\frac{2}{1-\beta})\E[Z_1] + \frac{\epsilon}{14\Delta}}
	  {1 - \frac{\epsilon}{14\Delta}}
      -
     \frac{(\frac{2}{1-\beta})\E[Z_3] - \frac{\epsilon}{14\Delta}}
	  {1 + \frac{\epsilon}{14\Delta}}
     \right)\\
& = &
   \tp\cdot v + 3\Delta \left[
     \left(\frac{2}{1-\beta}\E[Z_1] + \frac{\epsilon}{14\Delta}\right)\left(1 +
		\frac{\epsilon}{14\Delta}\right)\right. \\
 &  & \left. -
    \left(\frac{2}{1-\beta}\E[Z_3] - \frac{\epsilon}{14\Delta}\right)\left(1 -
		\frac{\epsilon}{14\Delta}\right) \right]
   \left(\frac{1}{1 - (\frac{\epsilon}{14\Delta})^2} \right)\\
& \leq &
   \tp\cdot v + 3.1\Delta \left[
     \left(\frac{2}{1-\beta}\E[Z_1] + \frac{\epsilon}{14\Delta}\right)\left(1 +
		\frac{\epsilon}{14\Delta}\right)\right. \\
& & \left. -
     \left(\frac{2}{1-\beta}\E[Z_3] - \frac{\epsilon}{14\Delta}\right)\left(1 -
		\frac{\epsilon}{14\Delta}\right)\right]\\
& = &
   \tp\cdot v + 3.1\Delta\left(
     \frac{2}{1-\beta}\left(\E[Z_1] + \E[Z_3]\right)
        \left(\frac{\epsilon}{14\Delta}\right) +
	\frac{\epsilon}{7\Delta}\right).\\
\end{eqnarray*}

Finally, using the fact that $Z_1+Z_3 \leq \sum_{i\in L}\tp_i =
1-\beta$, we get
\begin{eqnarray*}
\tp'\cdot v & \leq & \tp\cdot v +
3.1\Delta(\frac{\epsilon}{7\Delta} + \frac{\epsilon}{7\Delta}),
\end{eqnarray*}
yielding the desired bound $\tp'\cdot v \leq
\tp\cdot v + \epsilon$.  Applying the union bound over all $i\in
\{1, \ldots, k_1\}$ we obtain that the probability that there exists $v$ in ${v_1, \ldots, v_{k_1}}$ such that  $\tp' \cdot v \geq \tp \cdot v + \epsilon$  is at most $1/3$.

Consider an arbitrary distribution $p$ in $\{p_1, \ldots, p_{k_2}\}$. Assume that $p=\tp+g$. By definition, we have:

\begin{eqnarray*}
d(p,\tp') &=& \frac{1}{2}\sum_{i \in L}
{\left|\tp_i+{g}_i - \tp_i- \frac{3\Delta \tp_i X_i}{\sum_{i \in L}
\tp_i X_i } + \frac{3\Delta \tp_i (1-X_i)}{\sum_{i \in L}^{n} \tp_i
(1- X_i )} \right|}+ \frac{1}{2}\sum_{i \in H} {\left| {g}_i
\right|}\\
&=&  \frac{1}{2} \sum_{i \in L} {\left|{g}_i -
\frac{3\Delta \tp_i X_i}{\sum_{i \in L} \tp_i X_i } + \frac{3\Delta
\tp_i (1-X_i)}{\sum_{i \in L} \tp_i (1- X_i )} \right|} +\frac{1}{2}
\sum_{i \in H} {\left| {g}_i \right|}\\
&\geq& \frac{1}{2} \sum_{i \in L} {\left|{g}_i - \frac{6 \Delta \tp_i
X_i}{1-\beta}  +  \frac{6\Delta \tp_i (1-X_i)}{1-\beta}  \right|} -
\frac{1}{2} \sum_{i \in L}\left[{\frac{6 \Delta
\tp_i X_i}{5(1-\beta)}} + \frac{6 \Delta
\tp_i(1-X_i)}{7(1-\beta)}\right]
+ \frac{1}{2}\sum_{i \in H} {\left| {g}_i \right|}\\
&\geq & \frac{1}{2} \sum_{i \in L} {\left|{g}_i - \frac{6 \Delta \tp_i
X_i}{1-\beta}  +  \frac{6\Delta \tp_i (1-X_i)}{1-\beta}  \right|} -
\frac{1}{2}  \sum_{i \in L}{\frac{6 \Delta \tp_i}{5(1-\beta)}} +
\frac{1}{2}\sum_{i \in H} {\left| {g}_i \right|}\\
&=& \frac{1}{2} \sum_{i \in L} {\left|{g}_i - \frac{6 \Delta \tp_i
X_i}{1-\beta}  +  \frac{6\Delta \tp_i (1-X_i)}{1-\beta}  \right|} -
\frac{3\Delta}{5} + \frac{1}{2}\sum_{i \in H} {\left| {g}_i \right|}
\end{eqnarray*}
where the
first inequality follows from applying relation~(\ref{denom}) to the denominators, and the last equality follows from the fact that $\|\tp_L\|_1=1-\beta$.

Let us denote by $Z=\frac{1}{2} \sum_{i \in L} {\left|{g}_i - \frac{6 \Delta \tp_i X_i}{1-\beta}  +  \frac{6\Delta \tp_i (1-X_i)}{1-\beta}  \right|}$. We have:
$$E[Z] = \frac{1}{2} \sum_{i \in L} \left[ {\frac{1}{2}\left|{g}_i - \frac{6 \Delta \tp_i}{1-\beta}\right|   + \frac{1}{2} \left|{g}_i + \frac{6\Delta \tp_i}{1-\beta} \right|} \right] \geq \frac{1}{2} \sum_{i \in L} {\left|{g}_i\right|},$$ therefore

$$E[d(p,\tp')] \geq E[Z] +\frac{1}{2} \sum_{i \in H} {\left| {g}_i
\right|} -\frac{3\Delta}{5}  \geq \frac{1}{2} \sum_{i \in L}
{\left|{g}_i\right|} +\frac{1}{2} \sum_{i \in H} {\left| {g}_i
\right|} -\frac{3\Delta}{5}=d(p,\tp) -\frac{3\Delta}{5}.$$
We can now
apply McDiarmid's inequality (see Theorem~\ref{mcd}) to argue that
with high probability
$Z$  is within $2\Delta/5$ of its expectation.
Note that $c_i=  \frac{6 \Delta \tp_i}{1-\beta}.$ Therefore:
\begin{eqnarray*}
\Pr\left\{ \left|Z - \E[Z] \right| \geq 2\Delta/{5} \right\} \leq
2e^{-2\Delta^2(1-\beta)^2 /(225 \sum \limits_{i \in L}{\tp_i^2
\Delta^2})} \leq
%2e^{-\Omega(S)}
2e^{- (2/225)S  }
\leq  \frac{1}{3k_2}. \end{eqnarray*}

%
% $$E[d(p,\tp')]  = E[Z] +\frac{1}{2} \sum_{i \in H} {\left| {g}_i \right|} -\frac{\Delta}{2}  \geq \frac{1}{2} \sum_{i \in L} {\left|{g}_i\right|} +\frac{1}{2} \sum_{i \in H} {\left| {g}_i \right|} -\frac{\Delta}{2}=d(p,\tp) -\frac{\Delta}{2}.$$ We can now apply McDiarmid inequality (see Theorem~\ref{mcd}) to argue that  $Z$  is within a $\Delta /2$ of its expectation.
%Note that $c_i=  \frac{12 \Delta \tp_i}{1-\beta}.$ Therefore:
%\begin{eqnarray*}
%\Pr\left\{ \left|Z - \E[Z] \right| \geq \Delta/{2} \right\} \leq
%2e^{-(\Delta^2(1-\beta)^2) /72 \sum \limits_{i \in L}{\tp_i^2 \Delta^2}} \leq 2e^{- S  } \leq  \frac{1}{3k_2}. \end{eqnarray*}
This then implies that
\begin{eqnarray*}
\Pr\left\{ d(p,\tp') \leq d(p,\tp') -\Delta \right\} \leq  \frac{1}{3k_2}. \end{eqnarray*}
By the union bound we get that the probability that there exists a $p$ in  $\{p_1, \ldots, p_{k_2}\}$ such that  $d(p,\tp') \leq d(p,\tp') -\Delta$  is at most $1/3$.
Summing up overall all possible events we get that there is a non-zero probability of (1), (2), (3) happening, as desired.
\end{proof}

%\vspace{-3mm}

\begin{proofof}{Theorem~\ref{thm:weak:small-support}}
Let $(p^*,q^*)$ be an arbitrary Nash equilibrium.
We show that each of $p^*$ and $q^*$ are highly concentrated, meaning that all but at
most $8\Delta$ of their total probability mass is concentrated in a set of
size $O((\Delta/\epsilon)^2 \log(1 + 1/\Delta)\log n)$. Let's consider $p^*$ (the argument for $q^*$ is similar).
  We begin by
partitioning it into
its \emph{heavy} and \emph{light} parts.  Specifically, we greedily
remove the largest entries of $p^*$ and place them into a set $H$ (the
heavy elements) until either:
\begin{enumerate}
\item[(a)] the remaining entries $L$ (the light
elements) satisfy the condition that  $\forall i\in L$, $\Pr[i] \leq
\frac{1}{S} \Pr[L]$ for $S$ as in Lemma \ref{clm:p-small-norm}, or
\item [(b)] $\Pr[H] \geq 1-8\epsdist$,
\end{enumerate}
whichever comes first.  Using the fact that the game satisfies the well supported
$(\epsapx,\epsdist)$-approximation stability condition, we will show that case (a) cannot occur first, which will imply that $p^*$ is highly concentrated.

In the following, we denote $\beta$ as the total probability mass over $H$.
Assume by contradiction that case (a) occurs first.
%  then we show that the game cannot satisfy the well supported
%$(\epsapx,\epsdist)$-approximation stability.
Note that we have $\|p_L\|_1=1-\beta$, $\|p_H\|_1=\beta$, and
 $$\sum_{i \in L} (p_i)^2  \leq \frac{1} {S} \sum_{i  \in L} {p_i} \sum_{i  \in L} {p_i} = \frac{1}{S}(1-\beta)^2,$$ so $\|p_L\|_2^2 \leq \frac{(1-\beta)^2}{S}$.
Let $v_{i,j}=C(e_i-e_j)$. Since $(p^*,q^*)$ is a Nash equilibrium we know that
$p^* \cdot v_{i,j} \leq 0$ for all $i$ and for all $j \in \supp(q^*)$.

By Lemma~\ref{clm:p-small-norm} there exists  $\tp'$ such that (1) $d(p^*,\tp') = 3\Delta$,
(2) $\tp' \cdot v_{i,j} \leq p^* \cdot v_{i,j} + \epsilon$ for all $i$ and for all $j \in \supp(q^*)$ and
(3)  $d(\tp',p_i) > d(p^*,p_i)- \Delta$ for all $i \in \{1, \ldots, k\}$ (here $k$ is the number of equilibria of the game).
By (2) we have that $\tp' \cdot v_{i,j} \leq \epsilon$ for all $i$ and for all $j \in \supp(q^*)$, which implies that $(\tp',q^*)$
 is a well supported $\epsilon$ approximate equilibrium (since by (2) the column player has at most an $\epsilon$ incentive to deviate and since $\supp(\tp') \subseteq \supp(p^*)$ we know that the row player has no incentive to deviate).
By (1) we also have that $(\tp',q^*)$ is $3\Delta$-far from
$(p^*,q^*)$.  We now use (3) to show that $(\tp',q^*)$ is $\Delta$-far
from all the other equilibria as well. Let $p$ be such an
equilibria. Note that if $d(p,p^*) > 4\Delta$, then clearly, by the
triangle inequality $d(p,\tp') > \Delta$. If $d(p,p^*) <
2\Delta$, clearly, by the triangle inequality, $d(p,\tp') >
\Delta$. Finally if $d(p,p^*) \in [2\Delta,4\Delta]$, then by (3), we
that $d(p,\tp') > \Delta$, as desired.

Overall we get that $(\tp',q^*)$ is a well supported $\epsilon$ approximate equilibria that is $\Delta$-far from all the other equilibria of the game.
This contradicts  the well supported $(\eps,\epsdist)$-approximation stability condition, as desired.

Thus case (b) occurs first, which implies that $p^*$ is highly concentrated.  We clearly have $1-\beta \leq 8\epsdist$; moreover, it is easy to show that the set
$H$ has at most $S \log{(1 + (8\epsdist)^{-1})}$ elements.  The key idea is that since $1-\beta \leq 8\epsdist$, we can
now apply the sampling argument of~\cite{LMM03} to $L$ with accuracy parameter $O(\eps/\epsdist)$ and then
union the result with $H$.
Specifically, let us decompose $p^*$ as: $$p^* = \beta p_H + (1-\beta)
p_L.$$ Applying the sampling argument of~\cite{LMM03} to $p_L$, we
have that by sampling a multiset $\X$ of $S$ elements from $L=\supp(p_L)$,
we are guaranteed that for any column $e_j$, we have: $$\left| (U_{\X})^T Ce_j
- p_L^T C e_j\right| \leq (\eps/8\epsdist),$$ where $U_{\X}$ is the
uniform distribution over $\X$. This means that for
$\tilde p = \beta p_H + (1-\beta) U_S$, all columns $e_j$ satisfy:
$$|{p^*}^TCe_j - \tp^TCe_j| \leq \eps/2.$$
 We have thus found the row portion of an
$\eps$-equilibrium with support of size $S \log{(1 + (8\epsdist)^{-1})}$  as desired.
%%To complete the proof we apply the same argument to $q^*$.
\end{proofof}

\begin{corollary}\label{fixed}
Let us fix $\epsilon$ and $\Delta$, $0\leq \epsilon \leq \Delta \leq 1$.  Let  $G$ be a game with at most $n^{O((\epsdist/\epsapx)^2)}$ Nash equilibria satisfying  the well supported $(\epsapx, \epsdist)$-approximation stability  condition
  (or the $(\epsapx/2, \epsdist)$-stability to perturbations  condition).

$(1)$ Given $G$ we can find a well-supported $\epsapx$-equilibrium
$(p,q)$ of $G$ in time $n^{O((\epsdist/\epsapx)^2\logdel\log n)}$.

$(2)$ Given $G'$, an $L_\infty$ $\epsilon/6$-perturbation of $G$, we
can find a well supported $\epsapx$-equilibrium $(p,q)$ of $G$ in time $n^{O((\epsdist/\epsapx)^2\logdel\log n)}$.

In both cases,
%if G also satisfies the well supported $(3\epsapx, \epsdist)$-approximation stability  condition,
%then
$(p,q)$  is $\Delta$-close to a Nash equilibrium $(p^*,q^*)$ of $G$.
\end{corollary}

\begin{proof} (1)  By Theorem~\ref{thm:weak:small-support}, we can simply try all supports of size $n^{O((\epsdist/\epsapx)^2\logdel\log n)}$ and for each of them write an LP to search for a well-supported $\epsilon$-Nash equilibrium.

(2) By Theorem~\ref{thm:weak:small-support}, $G$ has a well-supported $\epsilon/3$-Nash equilibrium with support of size $O((\epsdist/\epsapx)^2\logdel\log n)$. Since $G'$ is an $L_\infty$ $\epsilon/6$-perturbation of $G$, then this is also a well-supported   $2\epsilon/3$-Nash equilibrium of $G'$. Thus by trying all supports of size $n^{O((\epsdist/\epsapx)^2\logdel\log n)}$ in $G'$ we can find
a well-supported   $2\epsilon/3$-Nash equilibrium of $G'$. Since $G$ is an $L_\infty$ $\epsilon/6$-perturbation of $G'$, this will be a well-supported   $\epsilon$-Nash equilibrium of $G$.
%%which by assumption. since $G$ is well supported $(\epsapx, \epsdist)$-approximation stable, will be $\Delta$-close to Nash equilibrium of $G$.
\end{proof}

Corollary~\ref{fixed} improves by a factor $O(1/(\epsdist^2\logdel))$ in the exponent over the
bound of \cite{LMM03} for games satisfying this condition.
% (and reducing to the bound of \cite{LMM03} in the worst case when
%$\epsdist=1$);
The most interesting range of improvements happens when $\epsilon$ is a function on $n$ and
$\Delta$ is a function of $\epsilon$; e.g., $\epsilon=1/\sqrt{n}$, $\Delta=10\epsilon$ --
in this case we obtain an improvement of $O(n/\log(n))$ in the exponent over the bound of \cite{LMM03}.

The proof of Theorem~\ref{thm:weak:small-support} also implies an interesting structural result, namely that {\em each} Nash equilibrium of
such a game is close to a pair of strategies of small support and by the triangle inequality, the same will happen for any perturbation of $G$. Formally:
\begin{theorem}
\label{thm:weak:small-support1}
Let us fix $\epsilon$ and $\Delta$, $0\leq \epsilon \leq \Delta \leq 1$.
Consider a game $G$ with at most $n^{O((\epsdist/\epsapx)^2)}$ Nash equilibria which satisfies
 the well supported $(\epsapx, \epsdist)$-approximation stability condition (or the  $(\epsapx/2, \epsdist)$-stability to perturbations  condition).
Then
it must be the case that:

$(1)$ Any Nash equilibrium  in $G$ is $8\Delta$-close to a pair of mixed strategies each with support of size at most
%$\epsapx$-equilibrium where each player's strategy has
{$O((\epsdist/\epsapx)^2\logdel\log n)$}.

$(2)$ For any game $G'$ with $L_\infty$ distance $\epsilon/2$ of
$G$, any Nash equilibrium  in $G'$ is $9\Delta$-close to a pair of mixed strategies each with support of size
%$\epsapx$-equilibrium where each player's strategy has
$O((\epsdist/\epsapx)^2\logdel\log n)$.
\end{theorem}

So far in Theorem~\ref{thm:weak:small-support} %%and~\ref{thm:weak:small-support1}
and Corollary~\ref{fixed} we have considered  $\epsilon$ and $\Delta$ fixed.
It is also interesting to consider games where the stability conditions hold uniformly for all $\epsilon$ small enough. We call such games uniformly stable games. Formally:
%
%Another particularly interesting class of games for which our work provides a
%  dramatic improvement are those that satisfy the the $(\epsilon,c\epsilon)$
%  perturbation-stability condition uniformly for all $\epsilon < \epsilon_0$, where $\epsilon_0=1/\poly(n)$.
%Formally:
\begin{definition} Consider $t \geq 1$.
We say that a game is {\em $t$-uniformly stable to perturbations} (or $t$-uniformly well supported approximation stable) if there exists $\epsilon_0 = 1/\poly(n)$ such
that for all $\epsilon \leq \epsilon_0$, $G$ satisfies
$(\epsilon,t\epsilon)$ stability to perturbations (or the well supported $(\epsilon,t\epsilon)$ approximation stability).
\end{definition}

  For games satisfying  the $t$-uniform stability to perturbations  condition with  $t=O(\poly(\log(n)))$ we can find
$1/\poly(n)$-approximate equilibria in $n^{\poly(\log(n))}$ time,  and more generally $\epsilon$-approximate equilibria in $n^{\log{(1/\epsilon)} \poly(\log(n))}$ time, thus achieving a FQPTAS.  This provide a dramatic improvement over the best worst-case bounds known.

\begin{corollary}\label{forall}
(1) Let $t=O(\poly(\log(n)))$.
There is a FQPTAS to find approximate-equilibria in games satisfying
 the $t$-uniform well supported approximation stability condition
  (or the $t$-uniform stability to perturbations  condition)
with at most $n^{O(t^2)}$ Nash equilibria.

(2) Games satisfying
 the $t$-uniform well supported approximation stability condition
  (or the $t$-uniform stability to perturbations  condition)
with at most $n^{O(t^2)}$ Nash equilibria have the property that for
 any $\Delta$ each equilibrium is $8\Delta$-close to a pair of mixed strategies
 each with support of size $O(t^2\logdel\log n)$; moreover, for any $L_\infty$-perturbation of magnitude $\Delta/t$ of such games,
it must be the case that any Nash equilibrium  in $G'$ is $9\Delta$-close to a pair of mixed strategies each with support of size
%$\epsapx$-equilibrium where each player's strategy has
$O(t^2\logdel\log n)$.
\end{corollary}

%A particularly interesting class of games for which this provides a drastic improvement over the previous work are those that satisfy
% what we call the
%  the $(\epsilon,c\epsilon)$
%perturbation-stability condition for all $\epsilon < \epsilon_0$, where $\epsilon_0=1/\poly(n)$  -- we say that such games satisfy the uniformly perturbation stability.
Corollary~\ref{forall} is especially interesting because the results of~\cite{cdt09} prove that it is PPAD-hard to find
$1/\poly(n)$-approximate equilibria in general bimatrix games.
  Our results shows that under the
(widely believed) assumption that PPAD is not contained in
quasi-polynomial time \cite{dask11},  such uniformly stable game are inherently easier for computation of approximate equilibria
than general bimatrix games.
%\footnote{We note that the generic Lipton et al.~\cite{LMM03} provides a FQPTAS
%only for  $\epsilon = \Omega(1/\poly(\log n))$.}
Moreover, variants of many games appearing commonly in experimental economics including the public goods game and identical interest game~\cite{gtbook1} %%%, ...{\tt INSERT ONE MORE.... %%rock-paper-scissors, matching pennies, or battle of sexes
satisfy this condition.

\section{Converting the general case to the stable case}
\label{lowerbounds}
In this section we show that
 computing a $\epsilon$-equilibrium in a game satisfying the strong $(\epsilon,\Theta(\epsilon^{1/4}))$ approximation stability
 is as hard as computing an $\Theta(\epsilon^{1/4})$-equilibrium in a general game.
 For our reduction, we show that any general game can be embedded into one having the strong
 $(\epsilon,\Theta(\epsilon^{1/4}))$ approximation stability  property such that an $\epsilon$ equilibrium in the new game yields an $\Theta(\epsilon^{1/4})$ in the original game.
 Since both notions of (strong) stability to perturbations  and (strong) well supported approximation stability
generalize the strong approximation stability condition, the main lower bound in this section (Theorem~\ref{mainlb}) applies to these notions as well.

We start by stating a useful lemma that shows the existence of a family  of modified matching pennies games that are strong approximation stable games with certain properties that will be helpful in proving our main lower bound. %% of Theorem~\ref{mainlb}.
%%We then present and prove our main lower bound (Theorem~\ref{mainlb}).
% and then prove
%Lemma~\ref{basic-game-stable}.

\begin{lemma}
\label{basic-game-stable}
Assume that $\Delta \leq 1/10$. Consider the  games defined by the matrices:
$$R=\left[\begin{array}{ccccc}
           1 + \alpha_{1,1} & 1 + \alpha_{1,2} & \ldots &   1 + \alpha_{1,n}  & 0 \\
              1 + \alpha_{2,1} & 1 + \alpha_{2,2} & \ldots &   1 + \alpha_{2,n}  & 0 \\
              & \ldots&  \ldots&  \\
           1 + \alpha_{n,1} & 1 + \alpha_{n,2} & \ldots &   1 + \alpha_{n,n}  & 0 \\
               0 & 0 & \ldots  & 0 &  2\Delta
\end{array}\right], ~~C=
\left[\begin{array}{ccccc}
          \gamma_{1,1} &  \gamma_{1,2} & \ldots &   \gamma_{1,n}  & 1 \\
          \gamma_{2,1} &  \gamma_{2,2} & \ldots &   \gamma_{2,n}  & 1 \\
              & \ldots&  \ldots&  \\
           \gamma_{n,1} & \gamma_{n,2} & \ldots &   \gamma_{n,n}  & 1 \\
               2\Delta & 2\Delta & \ldots  & 2\Delta & 0
\end{array}\right]$$
where  $\alpha_{i,j} \in [-\Delta, 0]$ and $\gamma_{i,j}\in [0, \Delta]$ for all $i,j$. Each such game satisfies the strong $(\Delta^2,4\Delta)$
approximation stability condition. Moreover if $(p,q)$ is a $\Delta^2$-Nash equilibrium, then we must have
$$\Delta/2 \leq p_1+...+p_n \leq 4 \Delta~~~\mathrm{ and }~~~\Delta/2 \leq q_1+...+q_n \leq 4 \Delta.$$
\end{lemma}

See Appendix~\ref{proofs} for a proof.
We now present the main result of this section.
\begin{theorem}
\label{mainlb}
Computing an $\epsilon$-equilibrium in a game satisfying  the strong $(\epsilon,8 \epsilon^{1/4})$ approximation stability condition
 is as hard as computing
an $(8\epsilon)^{1/4}$-equilibrium in a general game.
%Computing an $\epsilon$-Nash in a general game is as hard as computing an $\epsilon'^4$ well supported  Nash in a game
%satisfying $(\epsilon'^4,\epsilon')$ stability.
\end{theorem}
\begin{proof} The main idea is to construct a linear embedding of any given game into a larger game with one more strategy per player played with large probability, thereby compressing the incentives of the original game into a smaller scale.
In particular,  consider $\Delta=(8\epsilon)^{1/4}$ and consider a general game with payoff matrices $R$ and $C$. Let us construct a new game with the payoff matrix for the row
player  $R'$ defined as:

$$\left[\begin{array}{cc}
            \onem_{n,n} -  (\Delta/2) \onem_{n,n}  + (\Delta/2) R & \zerom_{n,1} \\
           \zerom_{1,n} &  2\Delta
\end{array}\right]$$

and for the column player $C'$ defined as:

$$\left[\begin{array}{cc}
           (\Delta/2) \onem_{n,n} + (\Delta/2) C & \onem_{n,1} \\
           2\Delta  \onem_{1,n}&  0
\end{array}\right]$$
(For $s,r >0$,  the matrix $\onem_{s,r}$ is the $s \times r$ matrix with all entries set to $1$ and  the matrix $\zerom_{s,r}$ is the $s \times r$ matrix with all entries set to $0$.)
 By Lemma~\ref{basic-game-stable}, the new game  defined by $R'$ and $C'$ satisfies the strong $(\Delta^2,4\Delta)$ approximation stability condition,
which in turn implies satisfying  $(\epsilon,8 \epsilon^{1/4})$ approximation stability (since $\epsilon \leq \Delta^2$ and $4 (8)^{1/4} \leq 8$).
 %% and so the  $(\Delta^4/8,3\Delta)$ approximation stability condition.
We  show next that any $\Delta^4/8$-equilibrium in this new game (defined by $R'$ and $C'$) induces  a $\Delta$-equilibrium in the
 original game (defined by $R$ and $C$). Since $\Delta=(8\epsilon)^{1/4}$, this implies the desired result.

Let $(p,q)$ be an $\Delta^4/8$-equilibrium in the new game. By Lemma~\ref{basic-game-stable} (since $\Delta^4/8 \leq \Delta^2$),
$p$ must have $\beta \Delta$ probability mass in the first $n$ rows and $q$ must have $\alpha \Delta$  probability mass in the first $n$ columns, where $\alpha, \beta \in [1/2, 4]$.
Let $p_{\fst}$, $q_{\fst}$ denote $p$ restricted to the first $n$ rows and $q$ restricted to the first $n$ columns.
Let $\tp_{\fst}=p_{\fst}/|p_{\fst}|$ and $\tilde{q}_{\fst}=q_{\fst}/|q_{\fst}|$, where $|p_{\fst}| = \beta \Delta$ and $|q_{\fst}| = \alpha \Delta$. We show that that $(\tp_{\fst}, \tilde{q}_{\fst})$ is a $\Delta$-equilibrium in the original game defined by $R$ and $C$.
We prove this by contradiction. Assume this is not the case. Assume first that the row player has an  $\Delta$  incentive to deviate.
 There must exist $e_i$ such that:
$$e_i^T R \tilde{q}_{\fst} > \tp_{\fst} R \tilde{q}_{\fst} + \Delta.$$
%Therefore
%$$e_i^T (\Delta/2) R \tilde{q_{\fst}} > \tilde{p_{\fst}} (\Delta/2) R \tilde{q_{\fst}} + \epsilon^4/2\Delta^2$$
Multiplying both sides by $\alpha \beta \Delta^3/2$ and using the fact that $\alpha \beta \geq 1/4$ we get:
\begin{eqnarray} \beta \Delta e_i^T (\Delta/2) R q_{\fst} > p_{\fst} (\Delta/2) R q_{\fst} +  \Delta^4/8. \label{rowineq1}\end{eqnarray}
We clearly have $\beta \Delta e_i^T (\onem_{n,n}- (\Delta/2) \onem_{n,n}) q_{\fst} = p_{\fst}^T (\onem_{n,n}- (\Delta/2) \onem_{n,n}) q_{\fst}$, and by adding this quantity as well as
 $p_{n+1} (2\Delta) q_{n+1}$ to both sides of inequality~\ref{rowineq1} we get:
\begin{eqnarray*}&& \beta \Delta e_i^T (\onem_{n,n}- (\Delta/2) \onem_{n,n} + (\Delta/2) R) q_{\fst} +  p_{n+1} (2\Delta) q_{n+1} > \\&& p_{\fst} (\onem_{n,n}- (\Delta/2) \onem_{n,n} + (\Delta/2) R) q_{\fst} + p_{n+1} (2\Delta) q_{n+1} + \Delta^4/8.\end{eqnarray*}
which implies:
%$$\Delta e_i^T (\onem_{n,n}- \Delta/2 \onem_{n,n} + \Delta/2 R) q_{\fst} + p_2 (2\Delta) q_2 > p_{\fst} (\onem_{n,n}- \Delta/2 \onem_{n,n} + \Delta/2) R q_{\fst} +p_2 (2\Delta) q_2 +  \epsilon^4/8.$$
% so
 $$(\beta \Delta e_i + p_{n+1} e_{n+1})^T R'q   > p^T R' q  +\Delta^4/8.$$
 Therefore there exists a deviation for the row player (namely moving all $\beta \Delta$ probability mass from rows $1,2, \ldots, n$ onto row $i$),
 yielding a benefit of $\Delta^4/8$ to the row player. This contradicts the assumption that  $(p,q)$ is an
  $\Delta^4/8$-equilibrium in the new game, as desired.

Assume now that the column player has an  $\Delta$  incentive to deviate.  There must exist $e_j$ such that:
$$\tp_{\fst}^T C e_j > \tp_{\fst}^T C \tilde{q}_{\fst} + \Delta.$$
Multiplying both sides by $\alpha \beta \Delta^3/2$ and using the fact that $\alpha \beta \geq 1/4$  we get:
$$p_{\fst}^T  (\Delta/2) C (\alpha \Delta e_j) > p_{\fst}^T (\Delta/2) C q_{\fst} +\Delta^4/8.$$
We have $p_{\fst}^T (\Delta/2) \onem_{n,n} \alpha \Delta e_j = p_{\fst}^T (\Delta/2) \onem_{n,n} q_{\fst}  $,  so:
\begin{eqnarray}p_{\fst}^T  ((\Delta/2) C +  (\Delta/2) \onem_{n,n}) \alpha \Delta e_j > p_{\fst}^T ( (\Delta/2) C +  (\Delta/2) \onem_{n,n}) q_{\fst} +\Delta^4/8. \label{columineq1}\end{eqnarray}
We also have  $ p_{n+1}(2\Delta, \ldots, 2\Delta) \alpha \Delta e_j=p_{n+1} (2\Delta, \ldots, 2\Delta )q_{\fst} $. By adding this quantity
as well as the term $p_{\fst} (1, \ldots, 1 )q_{n+1}$ to the both sides of the inequality~\ref{columineq1} inequality we get:
$$p^T  C' (\alpha \Delta e_j + q_{n+1} e_{n+1})   >p^T C' q + \Delta^4/8.$$
 Therefore there exists a deviation for the column player (namely moving all $\alpha \Delta$ probability mass from columns $1,2, \ldots, n$ onto column
 $i$), yielding a benefit of $\Delta^4/8$ to the column player.
This contradicts the assumption that  $(p,q)$ is an
  $\Delta^4/8$-equilibrium in the new game, as desired.
\end{proof}

%\vspace{-3mm}

Theorem~\ref{mainlb} implies that for any $\epsilon \leq (1/8) (0.3393)^4$, an algorithm for finding an $\epsilon$-equilibria
in a game satisfying the strong $(\epsilon,8 \epsilon^{1/4})$ approximation stability condition would imply a better than currently known
algorithm for finding approximate equilibria in  general games (with an approximation factor of $0.3393$).

\section{Stability in constant-sum games}
\label{zero-sum}
%%We provide here a characterization of stable constant-sum games.
Consider a game defined by $R$ and $C$.
Let $$\Ps^*=\{p, \exists~q~~s.t.~~(p,q)~~\mathrm{~is~a~Nash~equilibrium}\}$$ and $$\Qs^*=\{q, \exists~p~~s.t.~~(p,q)~~\mathrm{~is~a~Nash~equilibrium}\}.$$
We say that $p$ is $\Delta$-far from $\Ps^*$ if the minimum distance between $p$ and $p' \in \Ps^*$ is $> \Delta$.
 Let $v_R$ and $v_C$ be the unique values of the row and column player respectively in a Nash equilibrium~\cite{gtbook}.
Lemmas~\ref{zerosum1} and~\ref{zerosum2} below characterize constant sum
 games satisfying approximation stability terms of properties of the
 space of mixed strategies for the row player and column player
 separately.
Theorem~\ref{algzerosum} gives a polynomial time algorithm for determining the approximately best  parameters for the strong approximation
 stability property for a given game.

\begin{lemma}
\label{zerosum1}
If for any $p$ that is $\Delta$-far from $\Ps^*$ there exists $e_j$ such that $p^T R e_j < v_R-\alpha$ and for any q  that is $\Delta$-far from $\Qs^*$ there exists $e_j$ such that $e_j^T C q < v_C-\alpha$, then the game satisfies the $(\alpha/2,\Delta)$ approximation stability.
\end{lemma}

\begin{proof}
We show that any $(p,q)$ that is $\Delta$-far from all Nash equilibria cannot be an $\alpha/2$-equilibrium.
Consider $(p,q)$ that is $\Delta$-far from all Nash equilibria. Then either $p$  is $\Delta$-far from $\Ps^*$  or  $q$  is $\Delta$-far from $\Qs^*$.\footnote{This follows from the well known interchangeability property of constant-sum games, meaning that given two Nash equilibria points
$(p_1, q_1)$ and $(p_2, q_2)$, the strategy pairs $(p_1, q_2)$ and $(p_2, q_1)$ are also Nash equilibria. To see that note that if both $p$ and $q$ are close to $\Ps^*$ and $\Qs^*$ respectively, then there exists a Nash  equilibria $(p_1,q_1)$, $(p_2,q_2)$ such that $d(p,p_1) \leq \Delta$ and $d(q,q_1) \leq \Delta$. By interchangeability, we get that $(p_1,q_2)$ is a Nash equilibrium and we also have $d((p_1,q_2),(p,q))\leq \Delta$.}
Assume WLOG that $p$ is $\Delta$-far from $p^*$. We know that there exists $e_j$ such that $p^T R e_j < v_R-\alpha$. We show that $(p,q)$ cannot be an $\alpha/2$ Nash equilibrium. If $p^T R q < v_R-\alpha/2$, then this is not an $\alpha/2$-equilibrium since the row player could play its minimax optimal strategy and get $v_R$. On the other hand if $p^TRq \geq v_R-\alpha/2$, then $p^TCq \leq v_C+ \alpha/2$, but we know that $p^T C e_j > v_C+\alpha$, so the column player would have an $\alpha/2$ incentive to deviate, as desired.
\end{proof}

%\vspace{-3mm}

\begin{lemma}
\label{zerosum2}
If there exists $p$ that is $\Delta$-far from $\Ps^*$ such that $\min_{j}p^T R e_j \geq v_R-\alpha$ or if there exists $q$ that
is is $\Delta$-far from $\Qs^*$ such that $\min_{j}e_j^T C q \geq v_C-\alpha$, then the game cannot be $(\alpha,\Delta)$ approximation stable.
Moreover, if, in the former case, $\supp(p) \subseteq
\supp(p^*)$ for some $p^* \in \Ps^*$, or if, in the latter case,
$\supp(q)
\subseteq \supp(q^*)$ for some $q^* \in \Qs^*$, then the game cannot be
well-supported $(\alpha,\Delta)$ approximation stable.
\end{lemma}
\begin{proof}
Assume that there exists $p$ that is $\Delta$-far from $\Ps^*$ such that $\min_{j}p^T R e_j \geq v_R-\alpha$.
%%Let $p^*$ be the closest point to $p$ in $\Ps^*$and let $(p^*,q^*)$ be an equilibrium that certifies that $p^*$ that is in $\Ps^*$.
Let $p^* \in \Ps^*$ be such that $\supp(p) \subseteq \supp(p^*)$ if such
   $p^*$ exists, else let $p^* \in \Ps^*$ be arbitrary;   let $(p^*,q^*)$ be an equilibrium that certifies
  that $p^*$ that is in $\Ps^*$.
We know that ${p^*}^T R q^*=v_R$ and ${p^*}^T C q^*=v_C$. Clearly, $(p,q^*)$ that is $\Delta$-far from $(p^*,q^*)$. We show that $(p,q^*)$ is an $\alpha$-Nash equilibrium, i.e., neither player has more than an $\alpha$-incentive to deviate. We have $p^T R q^* \geq v_R -\alpha$ and ${p'}^T R q^*  \leq v_R$ for any $p'$  (since $q^*$ is minimax optimal)
%(otherwise $(p^*,q^*)$ would not be an equilibrium since the row player would want to deviate if ${p'}^T R q^*  > {p^*}^T R q^*$ ),
so the row player has at most $\alpha$-incentive to deviate. We also have $p^T C q^* \geq v_C$  (since $q^*$ is minimax optimal)
%(otherwise it would not an equilibrium since the row player would want to deviate, since $p^T R q^* > {p^*}^T R q^*$),
and the most the column player  could get is $v_C+\alpha$  since $\min_{j}p^T R e_j \geq v_R-\alpha$, so $\max_{j}p^T C e_j \leq v_C+\alpha$.
%(i.e., ${p'}^T R q^* \leq v_C +\alpha$).
\end{proof}

%\vspace{-3mm}

If the game satisfies the strong approximation stability condition, then  we can efficiently compute good approximations for the stability parameters. Specifically:
%Given game $(R,C)$ and value $\alpha$, want to find approximately minimum
%Delta such that game is $(\alpha,\Delta)$ stable (looking at strong-stability).

\begin{theorem}
\label{algzerosum}
Given any  $0<\alpha<1$, we can use Algorithm~\ref{zero-sumA} to whp determine $\Delta$ such that the game satisfies
the $(\alpha/2,2\Delta)$ strong approximation stability property, but not  $(\alpha,\Delta/2)$ strong approximation
 stability. The running time is polynomial $n^{O(1/\alpha^2)}$.
\end{theorem}
\begin{proof}
We first find a minimax optimal solution $(p^*,q^*)$ and then in Step 2, a small support $\alpha$-Nash $(p',q')$.
%We can do Step 2 efficiently since we have $(p^*,q^*)$.
In  Step 3 we find $\Delta$ such that all $\alpha/2$-Nash equilibria are within distance $\Delta$ of $(p',q')$ and there exists an $\alpha$-Nash equilibrium at distance $\Delta$ from $(p',q')$.
From the perspective of the row player, as shown in Lemma~\ref{zerosum1}, if there exists $e_j$ such that $p^T R e_j < v_R - \alpha$, then $(p,q)$ cannot be an $\alpha/2$ Nash equilibrium for any $q$, so all $\alpha/2$ Nash equilibria must satisfy $p^T R e_j \geq v_R - \alpha$ for all $j$. As shown in Lemma~\ref{zerosum2}, for any $p$ such that $p^T R e_j \geq v_R - \alpha$, for all $j$ we have that $(p,q^*)$ is an $\alpha$ Nash equilibrium. Similarly for the column player. So all $\alpha/2$-equilibria must be at distance at most $\Delta$ from $(p',q')$, where $\Delta$ is the output of  Algorithm~\ref{zero-sumA}, and there exists an $\alpha$-Nash equilibrium that is at distance at $\Delta$ from $(p',q')$ .
By triangle inequality, we obtain that the game is $(\alpha/2,2\Delta)$ stable and it is  {\em not} $(\alpha,\Delta/2)$ stable.
Note that the running time is polynomial since we we perform steps (A), (B) at most $n^{O(1/\alpha^2)}$ times, so overall the running time is polynomial $n^{O(1/\alpha^2)}$.
\end{proof}

\begin{algorithm}
 \caption{Determining the strong stability parameters of a constant sum game.}
  {\bf Input:}  $R$, $C$, parameter $\alpha$.
\begin{enumerate}
\item Solve for minimax optimal $(p^*,q^*)$.

\item Step 2: Apply the sampling procedure in~\cite{LMM03} from $(p^*,q^*)$ to get $(p',q')$ with support of
size $O((\log n)/\alpha^2)$ that is an $\alpha$-Nash.  Set $\Delta=0$.

\item Find  $\Delta$ as follows:%%the farthest $\alpha$-Nash from $(p',q')$ as follows:
\begin{enumerate}
  \item[(A)]  For each partition of the support of $p'$ into $\supp_{+}$ and $\supp_{-}$ do:
\begin{enumerate}
\item   Solve the following LP:
\begin{eqnarray*}
& & \max  \Delta ={\sum_{i \in \supp_{+}} {(p_i - p'_i)}
          + \sum_{i \in \supp_{-}} {(p'_i - p_i)}
          + \sum_{i \in \{1,2, \ldots, n\} \setminus \supp(p') } p_i}~~~ \\
& &   \mathrm{ s.t.}~~p_i \geq p'_i\mathrm{~~for~~all}~~i~~ \in \supp_{+} \\
& &    \, \,\, \, \,\, \, \,\,  p_i \leq p'_i \mathrm{~~for~~all~~}i~~~ \in \supp_{-} \\
& &   \, \,\,  \, \,\, \, \,\,  p^T R e_j \geq v_R - \alpha~~\mathrm{~~for~~all~~}j~~
\end{eqnarray*}
   \item  If  $\Delta$ is smaller than $v$, the value of the previous LP, reset $\Delta$ to be $v$.
\end{enumerate}
  \item[(B)]  for each partition of the support of $q'$ into $\supp_{+}$ and $\supp_{-}$ do:
\begin{enumerate}
\item   Solve the following LP:
\begin{eqnarray*}
& & \max \Delta = {\sum_{i \in \supp_{+}} {(q_i - q'_i)}
          + \sum_{i \in \supp_{-}} {(q'_i - q_i)}
          + \sum_{i \in \{1,2, \ldots, n\} \setminus \supp(q') } {q_i}} \\
& &  \mathrm{ s.t.}~~  q_i \geq q'_i\mathrm{~~for~~all~~}i~~ \in \supp_{+} \\
& &  \, \,\, \, \,\, \, \,\,      q_i \leq q'_i \mathrm{~~for~~all~~}i~~~ \in \supp_{-} \\
& &  \, \,\, \, \,\, \, \,\,      e_j^T C q \geq v_C - \alpha~~\mathrm{~~for~~all~~}j~~
\end{eqnarray*}
   \item  If  $\Delta$ is smaller than $v$, the value of the previous LP, reset $\Delta$ to be $v$.
\end{enumerate}

  \end{enumerate}
  \end{enumerate}
    {\bf Output:} Radius $\Delta$.
   \label{zero-sumA}
\end{algorithm}

%
%  So, the running time of Algorithm~\ref{zero-sumA} is polynomial if $\alpha$ is constant.
%  Intuitively, approximating the parameters for strong approximation
%  stability appears easier than for (non strong) stability, for the
%  following reason.  In order to prove that a game is {\em not
%  strongly stable}, it suffices to show there exist two far-apart
%  approximate equilibria (and to prove it is strongly stable it
%  suffices to show there exists an approximate equilibrium such that
%  all other approximate equilibria are close).  However, in the case
%  of non-strong stability, for the ``no'' case, simply exhibiting two
%  far-apart approximate equilibria is insufficient.  Instead, one must
%  show there is an approximate equilibrium that is far from the set of
%  {\em all} Nash equilibria.  This appears to be more difficult.

%\vspace{-5mm}

In order to determine the approximate  parameters for the strong well supported approximation
 stability property for a given game, we can adapt Algorithm~\ref{zero-sumA} as follows.
Given any  $0<\alpha<1$, we can use Algorithm~\ref{zero-sumA} to whp determine $\Delta_h$.
We then we re-run  Algorithm~\ref{zero-sumA}, but in the LP in (A) we
add the constraint that $p_i = 0$ for all $i \notin \supp(p^*)$ (and
similarly for $q$ in the LP (B)) to get a value $\Delta_l$. Then by lemmas~\ref{zerosum1} and~\ref{zerosum2},
we are guaranteed that the game satisfies the  $(\alpha/2, 2\Delta_h)$, but not the
 $(\alpha,\Delta_l/2)$ strong well supported  approximation stability property.

\subsection*{Acknowledgments}
We thank Avrim Blum, Dick Lipton, Yishay Mansour, Shanghua Teng, and Santosh Vempala for useful discussions.
We also thank Vangelis Markakis for pointing~\cite{Vangelis} to us.

This research was supported in part by NSF grant CCF-0953192, ONR grant
N00014-09-1-0751, and AFOSR grant FA9550-09-1-0538. This work was done in part while the first author was visiting Microsoft Research NE
and while the second author was a member of Microsoft Research NE.

%\vspace{-3mm}

{\small
\bibliography{paper-ec}
}

\appendix

\section{Standard Facts}
\label{facts}
 We start by stating the McDiarmid inequality
(see~\cite{DGL}) we use in our proofs:
\begin{theorem}
\label{mcd} Let $Y_1, ..., Y_n$ be independent random variables
taking values in some set $A$, and assume that $t: A^n \rightarrow
R$ satisfies:
$$\sup\limits_{y_1,...,y_n \in A, \overline{y}_i \in A} \left| {t(y_1, ..., y_n)- t(y_1, ...,y_{i-1}, \overline{y}_i, y_{i+1}, y_n)} \right| \leq c_i, $$
for all $i$, $1 \leq i\leq n $. Then for all $\gamma >0$ we have:
$$\Pr\left\{ \left|t(Y_1, ..., Y_n) - \E[t(Y_1, ..., Y_n)] \right| \geq \gamma\right\} \leq
2e^{-2\gamma^2 / \sum\limits_{i=1}^{n}{c_i^2}}$$
\end{theorem}

We now state a well known fact showing any pair of strategies that is sufficiently close to a Nash equilibrium is a sufficiently good approximate Nash equilibrium.
\begin{claim}
\label{claim1}
If $(p,q)$ is $\alpha$-close to a Nash equilibrium $(p^*,q^*)$ (i.e., if
$d((p,q),(p^*,q^*)) \leq \alpha$), then $(p,q)$ is a $3\alpha$-Nash
equilibrium.
\end{claim}
%\begin{proof}
%Define $\pmin_i = \min(p_i, p^*_i)$ and $\pmin = (\pmin_1,...,\pmin_n)$.
%   So, $p = \pmin + p'$, $p^* = \pmin + p''$, where $\sum_{i}p'_i \leq \alpha$ and $\sum_{i}p''_i \leq \alpha$.
% Similarly,  for $q$, $q^*$, define  $\qmin$, $q'$, $q''$.
%
%Let $v_R = {p^*}^T R q^*$ be the value to the row player in $(p^*,q^*)$.
%Since $(p^*,q^*)$ is Nash,
%we know $e_i^T R q^* \leq v_R$ for all rows $i$.
%We now show that the best response to $q$ is at most $v_R + \alpha$.
%To see this,  consider some row $e_i$. The expected payoff is
%         $$\sum_j q_j R_{i,j} = \sum_j (\qmin_j + q'_j) R_{i,j}
%                       \leq \sum_j (\qmin_j + q''_j)   R_{i,j} + \alpha
%                       \leq v_R + \alpha.$$
%
%The middle inequality holds because $\sum_j q'_j R_{i,j} \leq \alpha$
%(since $R_{i,j} \in [0,1]$), and the last because $\sum_j q''_j
%R_{i,j} \geq 0$.
%We now show that the expected payoff of $p$ against $q$ is at least
%$v_R - 2\alpha$.
%\begin{eqnarray*}p^T R q &=& \sum_{i,j} p_i q_j R_{i,j}
%          = \sum_{i,j} (p^*_i-p''_i+p'_i)(q^*_j-q''_j+q'_j)R_{i,j} \\
%          &=& v_R + (\sum_{i,j} p^*_i(-q''_j+q'_j)R_{i,j} + (\sum_{i,j} (-p''_i+p'_i)q_j R_{i,j}) \\
%        & \geq& v_R + (\sum_i p^*_i(-\alpha)) + (\sum_j (-\alpha)q_j)
%          = v_R - 2\alpha.
%          \end{eqnarray*}
%Thus, under $(p,q)$, the row player has no more than
%$3\alpha$ incentive to deviate, and we have the analogous argument
%for the column player.  So, $(p, q)$ is a $3\alpha$-equilibrium.
%\end{proof}

\section{Additional Proofs}
\label{proofs}
%%\subsection{ Proof of Lemma~\ref{basic-game-stable}.}
\noindent {\bf Lemma~~\ref{basic-game-stable}}
Assume that $\Delta \leq 1/10$. Consider the games defined by the matrices:
$$R=\left[\begin{array}{ccccc}
           1 + \alpha_{1,1} & 1 + \alpha_{1,2} & \ldots &   1 + \alpha_{1,n}  & 0 \\
              1 + \alpha_{2,1} & 1 + \alpha_{2,2} & \ldots &   1 + \alpha_{2,n}  & 0 \\
              & \ldots&  \ldots&  \\
           1 + \alpha_{n,1} & 1 + \alpha_{n,2} & \ldots &   1 + \alpha_{n,n}  & 0 \\
               0 & 0 & \ldots  & 0 &  2\Delta
\end{array}\right], ~~C=
\left[\begin{array}{ccccc}
          \gamma_{1,1} &  \gamma_{1,2} & \ldots &   \gamma_{1,n}  & 1 \\
          \gamma_{2,1} &  \gamma_{2,2} & \ldots &   \gamma_{2,n}  & 1 \\
              & \ldots&  \ldots&  \\
           \gamma_{n,1} & \gamma_{n,2} & \ldots &   \gamma_{n,n}  & 1 \\
               2\Delta & 2\Delta & \ldots  & 2\Delta & 0
\end{array}\right]$$
where  $\alpha_{i,j} \in [-\Delta, 0]$ and $\gamma_{i,j}\in [0, \Delta]$ for all $i,j$. This game satisfies the strong $(\Delta^2,4\Delta)$
approximation stability condition. Moreover if $(p,q)$ is a $\Delta^2$-Nash equilibrium, then we must have
$$\Delta/2 \leq p_1+...+p_n \leq 4 \Delta~~~\mathrm{ and }~~~\Delta/2 \leq q_1+...+q_n \leq 4 \Delta.$$

\begin{proof} %{Lemma~\ref{basic-game-stable}}
First note that $e_{n+1}^T R q= 2\Delta q_{n+1}$ and
$$e_{i}^T R q=  (1+\alpha_{i,1}) q_1 + (1+\alpha_{i,2})q_2 + \ldots + (1+\alpha_{i,n})  q_n~~\mathrm{~~for~~}~~1 \leq i \leq n.$$
Also $p^T C e_{n+1}=  p_1 + \ldots + p_n$ and
$$p^T C e_{j}= p_1 \gamma_{1,j} + \ldots + p_n \gamma_{n,j} + 2 \Delta p_{n+1}~~\mathrm{~~for~~}~~1 \leq j \leq n.$$

By a simple case analysis, one can show that any Nash equilibrium $(p,q)$ must have $0<p_{n+1} < 1$ and $0<q_{n+1} <1$.
(This is also implicit in our analysis on  $\Delta^2$-Nash equilibria below).
This then implies that in any Nash equilibrium $(p,q)$ such that $p_i \neq 0$ we must have:
\begin{eqnarray}2\Delta q_{n+1}=(1+\alpha_{i,1}) q_1 + (1+\alpha_{i,2})q_2 + \ldots + (1+\alpha_{i,n})  q_n.\label{exactnashone}\end{eqnarray}
Similarly, in any Nash equilibrium $(p,q)$ such that $q_j \neq 0$ we must have:
\begin{eqnarray}p_1 + \ldots p_n=  p_1 \gamma_{1,j} + \ldots p_n \gamma_{n,j} + 2 \Delta p_{n+1}.\label{exactnashtwo}\end{eqnarray}

Identities~\ref{exactnashone} and~\ref{exactnashtwo} together with the fact that  $\alpha_{i,j} \in [-\Delta, 0]$ and $\gamma_{i,j}\in [0, \Delta]$
 for all $i,j$, imply that there must exist a Nash equilibrium  $(p,q)$ satisfying:
\begin{eqnarray} \frac{2\Delta}{1+2\Delta} \leq p_1+...+p_n \leq \frac{2\Delta}{1+\Delta}~~~\mathrm{and}~~~
 \frac{2\Delta}{1+2\Delta} \leq q_1+...+q_n\leq \frac{2\Delta}{1+\Delta}. \label{five}\end{eqnarray}

To get the desired stability guarantee we now show that any $\Delta^2$-equilibrium  must have $\Delta/2 \leq p_1+...+p_n \leq 4 \Delta$ and
$\Delta/2 \leq q_1+...+q_n \leq 4 \Delta$. (This in turn implies that any $\Delta^2$-equilibrium  must be at distance at most $4\Delta$ from a Nash equilibrium satisfying relation (\ref{five}).)
We prove this by contradiction.
Consider an arbitrary $\Delta^2$-equilibrium $(p,q)$. We analyze a few cases.

Case $1$: Suppose $p_{n+1} > 1 - \Delta/2$.  Then the column player's payoff for column $n+1$ is $p^T C e_{n+1} = \sum_{i=1}^{n}{p_i} \leq \Delta/2.$
  But the column player's payoff  for a  column  $j \in \{1,...,n\}$ is:
  $$ p^T C e_{j} =\sum_{i=1}^n \gamma_{i,j} p_i + 2 \Delta p_{n+1} \geq 2\Delta(1-\Delta/2).$$
  If $q_{n+1} > 1/2$ then the column player has incentive to deviate at least:
      $$p^T C e_{j} - p^T C q \geq (1/2)[2\Delta(1-\Delta/2) - \Delta/2] > \Delta/2 > \Delta^2,$$
      which cannot happen since $(p,q)$ is a $\Delta^2$-equilibrium.
  On the other hand if $q_{n+1} \leq 1/2$, then the row player has huge incentive to deviate.
  Specifically, the row's player payoff for row $1$ is $e_1^T R q \geq   (1/2) (1-\Delta)$, but row's player payoff for  row $n+1$ is
 $e_{n+1}^T R q \leq (1/2) 2\Delta =\Delta$. Thus in this case the row player has incentive to deviate at least:
  $$e_1^T R q -p^T R q \geq (1-\Delta/2) [1/2 (1-\Delta)- \Delta ]  > \Delta,$$ which cannot happen since $(p,q)$ is a $\Delta^2$-equilibrium.

Case $2$: Suppose $p_{n+1} < 1 - 4 \Delta$.
  Then the column player's payoff for column $n+1$ is $p^T C e_{n+1} =\sum_{i=1}^{n} p_i \geq 4\Delta$, whereas the column player's payoff
  for a  column  $j \in \{1,...,n\}$ is:
   $$p^T C e_{j}= p_1 \gamma_{1,j} + \ldots p_n \gamma_{n,j} + 2 \Delta p_{n+1}  \leq \Delta \sum_{i=1}^{n} p_i + 2\Delta p_{n+1} = \Delta (1-p_{n+1})+ 2 \Delta p_{n+1} \leq 2\Delta.$$
    So, if $q_{n+1} < 1 - \Delta/2$, then the
  column player has incentive to deviate at
  least: $$p^T C e_{n+1} - p^T C q \geq (\Delta/2) [4\Delta - 2\Delta] \geq \Delta^2,$$ contradiction.
  On the other hand, if $q_{n+1} > 1 - \Delta/2$, then the row player's payoff for
  row $n+1$ is $e_{n+1}^T R q \geq 2\Delta (1-\Delta/2),$ but the row player's payoff for a rows $i \in \{1,...,n\}$
  is: $$e_{i}^T R q=  \sum_{j=1}^{n}(1+\alpha_{i,j}) q_j \leq 1- q_{n+1} \leq  \Delta/2.$$
  So, in this case, the row player has incentive to deviate at least:
$$e_{n+1}^T R q - p^T R q \geq 4 \Delta [2 \Delta (1-\Delta/2)  - \Delta/2] > \Delta^2,$$
which cannot happen since $(p,q)$ is a $\Delta^2$-equilibrium.

Case $3$: Suppose $q_{n+1} > 1 - \Delta/2$. As in the bottom-half of the  case $2$ analysis we have that the
    row player's payoff for row $n+1$ is $e_{n+1}^T R q \geq  2\Delta (1-\Delta/2)$, but the
    row player's payoffs for rows $1,...,n$ are $\leq \Delta/2$.
  So, if $p_{n+1} < 1 - \Delta$ then the row player has incentive to deviate at least:
    $$e_{n+1}^T R q - p^T R q \geq \Delta  [2\Delta (1-\Delta/2)  - \Delta/2] > \Delta^2,$$
 which cannot happen since $(p,q)$ is a $\Delta^2$-equilibrium.
 On the other hand,  if $p_{n+1} \geq 1-\Delta$, then the column player's payoff for column $n+1$ is $p^T C e_{n+1} =\sum_{i=1}^{n} p_i \leq \Delta$,
  but the column player's payoff for a columns $j\in \{1,...,n\}$ is:
  $$p^T C e_{j}= p_1 \gamma_{1,j} + \ldots p_n \gamma_{n,j} + 2 \Delta p_{n+1}  \geq (1-\Delta)(2\Delta).$$
  So, the column player has incentive to deviate at least:
    $$p^T C e_{j} - p^T C q \geq q_{n+1}[2\Delta (1-\Delta) - \Delta] \geq \Delta/2 > \Delta^2,$$
 which cannot happen since $(p,q)$ is a $\Delta^2$-equilibrium.

Case $4$: Finally assume that $q_{n+1} < 1 - 4\Delta$.
  Then the row player's payoff  for row n+1 is $ e_{n+1}^T R q  \leq 2\Delta (1-4\Delta)$, but
row player's payoffs for rows $1,...,n$ are $\geq (4\Delta)(1-\Delta)$.
  So, if $p_{n+1} > 1/2$ then Row has incentive to deviate at least:
  $$ e_{n+1}^T R q -p^T R q\geq  (1/2)[4\Delta (1-\Delta) - 2\Delta (1-4\Delta)  ]\geq \Delta/2 > \Delta^2,$$
  which cannot happen since $(p,q)$ is a $\Delta^2$-equilibrium.
  Finally if $p_{n+1} < 1/2 < 1- 4\Delta$ we apply the analysis in case $2$.

Thus any $\Delta^2$-equilibrium  must have $\Delta/2 \leq p_1+...+p_n \leq 4 \Delta$ and
$\Delta/2 \leq q_1+...+q_n \leq 4 \Delta$, as desired.
This concludes the proof.
\end{proof}

\begin{lemma}
\label{parameters}
Assume that the game $\game$ satisfies the $(\epsapx,
\epsdist)$-approximation stability and that the union of all $\Delta$-balls around all Nash equilibria do not cover the whole space.
Then we must have  $ 3\epsdist \geq  \epsapx$.
\end{lemma}
\begin{proof}
Since the union of all $\Delta$-balls around all Nash equilibria do not cover the whole space, we must have a $(p,q)$ that is at distance exactly
$\Delta$ from some fixed Nash  equilibrium and that is $\Delta$-far from all the other Nash equilibria. By Claim~\ref{claim1} we also have that this is a $3\Delta$ Nash equilibrium. This then implies the desired result.
\end{proof}

\begin{lemma}
\label{parameterswellsupported}
 Assume that the bimatrix game $\game$ specified by $R$ and $C$ has a non-pure Nash equilibrium.
\begin{enumerate}
\item [(a)] If $\game$ satisfies the strong well supported $(\epsapx, \epsdist)$-approximation stability condition, then
  we must have  $ \epsdist \geq  \epsapx/4$.

\item [(b)] If $\game$ satisfies the strong $(\epsapx, \epsdist)$-stability to perturbations condition, then
  we must have  $ \epsdist \geq  \epsapx/8$.
  \end{enumerate}
\end{lemma}

\begin{proof}
Assume $\game$ satisfies the strong well supported $(\epsapx, \epsdist)$-approximation stability condition.
By definition, there exists a Nash equilibrium $(p^*,q^*)$ such that any
  well supported $\epsapx$-equilibrium is $\epsdist$-close to $(p^*,q^*)$.
Let $(p,q)$ be an arbitrary non-pure Nash equilibrium of $\game$ and assume WLOG that $p$ is a mixed strategy.
Consider an $\alpha$ internal deviation of the row player, i.e., consider $p'$ with $\supp(p')\subseteq \supp(p)$ such that $d(p,p')=\alpha$ .
 %Then $(p',q)$ is a $2\alpha$ well supported Nash equilibrium.
 Since $\supp(p')\subseteq \supp(p)$ we have $p'^T R q =p^T R q$. Since $(p,q)$ is a Nash equilibrium  we have
 $p^T C e_j = p^T C q \equiv v_C$ for all $j \in \supp(q)$ and $p^T C e_j \leq v_C$ for all $j \notin \supp(q)$.
Since $d(p,p')=\alpha$ we have $$|p'^T C e_j - p^T C e_j| \leq |(p' - p) ^T C e_j| \leq \alpha, $$ for all $j$, so
$p'^T C e_j \geq v_C -\alpha$,  for all $j \in \supp(q)$ and $p'^T C e_j \leq v_C + \alpha$,  for all $j \notin \supp(q)$.
Thus $(p',q)$  is a well supported $2\alpha$-Nash equilibrium. By construction, we have $d((p',q),(p,q))=\alpha$.
Since $d$ is a metric, by the triangle inequality, we get that at least one of the pairs $(p',q)$ and $(p,q)$  is
at least $\alpha/2$ far from $(p^*,q^*)$; however they are both $2\alpha$ well supported Nash equilibria.
 This implies that we must have $\Delta \geq \epsilon/4$, as desired.
By Theorem~\ref{stab3eqstab2}, we immediately get (b) as well.
\end{proof}

\section{Examples}
\label{simple-examples}
%To illustrate our notions of stability, in this section we give several simple examples of $2\times 2$ games.
%The first example shows that a game may be $(x,0)$ well-supported approximation stable. At the
%same time, it is easy to see that no game can be $(x,y)$ approximation stable for $x>y$ since
%the maximum penalty for deviating by $y$ units from the Nash equilibrium is $y$. The first example
%already shows that well-supported approximation stability is strictly weaker than the
%general approximation stability. The second example provides an even more dramatic demonstration of this
%fact by showing a very big gap between the well-supported approximation stability and the ordinary approximation stability
%that the game satisfies. The third example is borrowed from \cite{Vangelis}, demonstrates that even the weaker
%stability condition is not redundant: it gives an example of a game that
%is not $(\epsapx,1-\delta)$-stable to perturbation, and thus also violates well-supported approximation stability.
To illustrate our notions of stability, we  present two $n$ by $n$ games satisfying  uniform stability to perturbations. %% (examples $4$ and $5$).

{\bf Example 1~~} A classic game from experimental economics is the public goods game which is defined is as follows. We have two players and each can choose to play a number between $0$ and $n-1$ corresponding to an amount of money to contribute. If the Row  player contributes $i$ dollars and the Column player contributes $j$ dollars, then each gets back $0.75(i+j)$.  %%I am assuming we multiply by 1.5 what both players put down and divide by 2.
So the payoff to the Row player is $0.75(i+j)-i$ and the payoff to the Column player is $0.75(i+j)-j$,
where $i \in \{0,1,...,n-1\}$ and $j \in \{0,1,...,n-1\}$.
This has payoffs ranging from $0$ up to $0.75(n-1)$, so to scale to the range $[0,1]$ as we do in our paper, we multiply all the payoffs by $1/n$.  I.e., if the Row player plays i and the Column player plays j then the payoff to the Row player is $[0.75j-0.25i]/n$ and the payoff to the Column player is
$[0.75 i-0.25 j]/n$.

First note that this game is $(\epsilon,0)$ stable to perturbations for all $\epsilon < 1/(8n)$.
To see this note that without any perturbation, for any $j$ and any $i \geq 1$ we have
         $e_0^T R e_j - e_i^T R e_j = 0.25 i/n \geq 0.25/n.$
That means that the Row  player prefers playing action $0$ compared to action $i$ by $0.25 i/n \geq 0.25/n$. So, in a game $R', C'$ that is an $L_\infty$ $\epsilon$-perturbation of of our game we get:
         $e_0^T R' e_j - e_i^T R' e_j \geq  0.25/n - 2\epsilon > 0.$
That means that in the perturbed game, the Row  player  still prefers playing action $0$.
This implies that the only equilibrium in the perturbed game has the  Row player playing action $0$, and similarly the Column player playing $0$, so the only equilibrium is $(0,0)$.

We now claim that  this game is not $(\epsilon,0.99)$ stable for any $\epsilon > 1/(4n)$.
To see this consider adding $\epsilon$ to $R[1,0]$.  I.e., if the  Row player plays action $1$ and the Column player plays action $0$, then the payoff to Row player is $-0.25/n + \epsilon > 0$. Now, $(1,0)$ is a Nash equilibrium since this payoff is strictly greater than $R[i,0]$ for any $i \neq 1$.  In particular, $R[0,0]=0$ and $R[i,0] < 0$ for all $i \geq 2$. So, there is now a Nash equilibrium (actually the unique Nash equilibium) at variation distance $1$ from the original Nash
equilibrium.

{\bf Example 2~~} We present here a variant of the identical interest game. Both players have  $n$ available actions. The first action is to stay home, and the other actions  correspond to $n-1$ different possible meeting locations. If a player chooses action $1$ (stay home), his payoff is $1/2$ no matter what the other player is doing. If the player chooses to go out to a meeting location, his payoff is $1$ if the other player is there as well and  it is $0$ otherwise. %%We claim that this game is the game is $(\epsilon,2\epsilon)$-stable for all $\epsilon < 1/6$.
Formally, $R[1,j] = 1/2$ for all $j$,          $R[i,i] = 1$ for all $i > 1$
          $R[i,j] = 0$ for $i > 1$, $j \neq i$.
and similarly $C[i,1] = 1/2$, $C[j,j]=1$ for $j > 1$, $C[i,j]=0$ for $j>1,i \neq j$.
We claim that this  game is well supported $(\epsilon,2\epsilon)$-stable for all $\epsilon < 1/6$, so it is $2$-uniformly stable.

Note that $e_1^T R q=1/2$ and $e_i^T R q=q_i$ for $i>1$. Similarly, $p^T C e_1=1/2$ and $p^T C e_i=p_i$ for $i>1$.
Note that if $(p,q)$ is an  well supported $\epsilon$-Nash equilibrium and if $q_i < 1/2 - \epsilon$ for $i > 1$
then both $p_i=0$ and $q_i=0$. This follows immediately since  $e_1^T R q=1/2$  and  $e_i^T R q=q_i$ for $i>1$,
so $p_i$ must equal $0$ on any action
whose expected payoff is $ < 1/2 - \epsilon$.  Since $p_i = 0$, $q_i$ must equal $0$ as well in order to be well-supported.
Also note that  if $(p,q)$ is an  well supported $\epsilon$-Nash equilibrium and if $q_i > 1/2 + \epsilon$ for $i > 0$
then both $p_i=1$ and $q_i=1$. If $q_i > 1/2 + \epsilon$ we have $e_i^T R q=1/2 + \epsilon$ and since $e_j^T R q \leq 1/2$ for $j \neq i$ we must have $p_i=1$. This in turn implies $q_i=1$.

Similarly, we can show the same for the row player as well.
These imply that the  well supported $\epsilon$-Nash equilibria  that are not already Nash equilibria must satisfy: for any action $i>1$, $i \in \supp(q)$, we have $1/2 - \epsilon \leq q_i \leq 1/2 + \epsilon$  and  $1/2 - \epsilon \leq p_i \leq 1/2 + \epsilon$. Similarly, for any action $i>1$, $i \in \supp(p)$, we have $1/2 - \epsilon \leq q_i \leq 1/2 + \epsilon$  and  $1/2 - \epsilon \leq p_i \leq 1/2 + \epsilon$.
We have two cases. The first one is if there is exactly one action $i > 1$ in $\supp(q)$.  In that case, $(p,q)$ has
distance at most $\epsilon$ from the  Nash equilibrium $(1/2 e_0 + 1/2 e_i, 1/2 e_0 + 1/2 e_i)$.
The second one is if there are two such actions $i,j > 0$ in $\supp(q)$.  In that case, $(p,q)$ has
distance at most $2\epsilon$ from the  Nash equilibrium
$(1/2 e_i + 1/2 e_j, 1/2 e_i + 1/2 e_j)$, as desired.
%
%[worst case is if p = 2*epsilon e_0 + (1/2 - epsilon)e_i + (1/2 - epsilon)e_j]

\end{document}